\documentclass{article}
\usepackage{PRIMEarxiv}

\usepackage[utf8]{inputenc} 
\usepackage[T1]{fontenc}    
\usepackage{hyperref}       
\usepackage{url}            
\usepackage{booktabs}       
\usepackage{amsfonts}       
\usepackage{nicefrac}       
\usepackage{microtype}      
\usepackage{lipsum}
\usepackage{fancyhdr}       
\usepackage{graphicx}       
\graphicspath{{media/}}     
\usepackage{amsmath,amssymb,amsfonts,amsthm}
\usepackage{algorithmic}
\usepackage{textcomp}
\usepackage{xcolor}
\usepackage{blindtext}
\usepackage{subfig}
\usepackage{cuted}
\usepackage{physics}
\usepackage{xcolor}
\usepackage[long]{optidef}
\usepackage[detect-none]{siunitx}
\usepackage{multirow}
\usepackage{array}
\usepackage{makecell}
\usepackage{textcomp}
\usepackage{algorithm}
\usepackage{bm}
\usepackage{epstopdf}
\newtheorem{theorem}{Theorem}[section]

\newtheorem{proposition}[theorem]{Proposition}

\newtheorem{definition}[theorem]{Definition}

\newtheorem{example}[theorem]{Example}

\newtheorem{result}[theorem]{Result}
\newtheorem{remark}[theorem]{Remark}
\title{Quasi-cyclic Linear Error-Block Code-based Post-quantum Signature
\thanks{\textit{This material is based upon work supported by Higher Education Authority (HEA), Ireland. \\
\copyright 2025 IEEE. Personal use of this material is permitted. Permission from IEEE must be obtained for all other uses, in any current or future media, including reprinting/republishing this material for advertising or promotional purposes, creating new collective works, for resale or redistribution to servers or lists, or reuse of any copyrighted component of this work in other works.}} 
}

\author{I. Cherkaoui, S. Belabssir, J. Horgan and I. Dey}

\begin{document}
\maketitle

\begin{abstract}
Shor’s quantum algorithm led to the discovery of multiple vulnerabilities in a number of cryptosystems. As a result, post-quantum cryptography attempts to provide cryptographic solutions that can face these attacks, ensuring the security of sensitive data in a future where quantum computers are assumed to exist. Error correcting codes are a source for efficiency when it comes to signatures, especially random ones described in this paper, being quantum-resistant and reaching the Gilbert-Varshamov bound, thus offering a good trade-off between rate and distance. In the light of this discussion, we introduce a signature based on a family of linear error-block codes (LEB), with strong algebraic properties: it is the family of quasi-cyclic LEB codes that we do define algebraically during this work.
\end{abstract}


\section{Introduction}
\label{sec:introduction}
Nowadays, much cryptography relies on error-correcting codes and lattices. In particular, digital signature algorithms attract the attention of several researchers when compared with other schemes of cryptography. This is not only due to the interesting algebraic and arithmetic properties of linear error-correcting codes, but also since Shor's theorem \cite{Shor97} showed that with the advent of quantum computers, cryptosystems based on some theoretical theorems can be broken in polynomial time. This makes it very important to have proper, provably secure post-quantum signature schemes.\\

A great number of schemes based on the theory of coding have been advanced over the years. However, a lot of them were broken, while the secure ones suffered either from huge public keys, large signatures, or slow signing algorithms.\\
The first paper describing proposed a code-based signature scheme in the Lee metric was Fuleeca \cite{fuleeca} with an aim to providing resistance against quantum attack while at the same time offering small key/signature sizes. Unlike conventional code-based schemes (for instance, CFS, Wave) that mention Hamming or rank metrics, Fuleeca works with the Lee metric. The Lee metric is Manhattan-like distance over integers modulo p and complicates attacks that rely on sparsity or low supports because low Lee-weight vectors may correspond to large Hamming weights. Also underpinning security claims Fuleeca made was the NP-hard decoding problem for the Lee metric \cite{weger}. It implemented the hash-and-sign paradigm, using quasi-cyclic codes. The secret key was a generator matrix with circulant blocks, while the public key was its systematic form. Signatures were codewords with a controlled Lee weight such that a large sign match with the hashed message binds signatures to messages. Fuleeca has claimed 1,318-byte public keys and 1,100-byte signatures with NIST Level I being their performance criterion, which translates to total communication costs (public key + signature) lower than those for Dilithium and SPHINCS+. 

Rejecting sampling, as were Falcon/Wave, was avoided in the designs, and limits were put in place by use of a heuristic concentration step to further reduce signature leakage. In doing so, the following paper \cite{fuleakage}indicated three fatal attacks that bring down Fuleeca's back-security claims to untenable levels. Leaked-Sublattice Attack: The signatures lie in a low-dimensional sublattice spanned by the generator matrix rather than in the full code lattice so an attacker could recover secret vectors via BKZ lattice reduction, having reconstructed the sublattice from almost 100 signatures. This reduced security for Fuleeca-I/III/V from 160/224/288 bits down to 111/155/199 bits, respectively. Worst-Case Lattice Attack (Quantum): The generator matrix having a circulant structure is mapped to an ideal lattice. Quantum algorithms \cite{biasse} will recover short generators of principal ideals in polynomial time. In quantum polynomial time with only a few signatures, full key recovery is possible. The learning attack: The concentration step put a bias into signature coefficients, leaking the secret-key structure. By averaging the outer products of the signatures, the attackers constructed an approximation of the secret circulant vector and refined it with rounding. Full key recovery with 90,000-175,000 signatures (practical for NIST Level I/V).\\

Thus, to mitigate the following issues due to heuristic defenses the choice of our proposed signature is based on LEB quasi-cyclic codes and a $\pi$ metric allowing the use of their algebraic properties and makes the scheme viable in practice.\\

Linear error-block codes or LEB codes for short, were introduced by Feng, Xu and Hickernell in \cite{FXH06} to be a generalization of classical linear codes. Feng \textit{ et al.} \cite{FXH06} have claimed that LEB codes have applications in experimental design since they yield mixed-level orthogonal arrays, and in high-dimensional numerical integration. In  order  to  allow  their application  in  cryptography,  especially  in  a  McEliece-like  cryptosystem,  Dariti \textit{et al.} \cite{DSa12b} presented  a  method  for  decoding  linear  error-block  codes  inspired  from  the  standard  array  classical   method. The same authors presented in \cite{Dariti12} some solutions on the use of LEB codes in codes-based cryptosystems, namely, the McEliece-like and Niederreiter cryptosystems, and realized that these solutions keep the size of the public key unchanged while they preserve, or even enhance, security parameters of the cryptosystem. They have also proposed a generalization to the Courtois-Finiasz-Sendrier (CFS) digital signature scheme \cite{Dariti12}, which may increase the efficiency of this scheme. In particular, the LEB codes allow the number of attempts required to find a decodable hash to be reduced.  In the same work \cite{Dariti12}, a  channel  model which  enables  LEB  codes  to  be  used  in  correcting  errors  raised  from  transmission  over  a noisy channel was designed. LEB codes have application in the field of steganography, where Dariti \textit{et al. } \cite{Dsa11} have introduced a protocol of steganography based on LEB codes. They have shown that employing convenient codes enhances the reliability of this protocol compared to other known steganography protocols. The method developed by the authors generalizes the classic idea of Crandall, and allows more data from the cover object to be exploited. They took the example of hiding information in a greyscale image. The proposed method makes it possible to use not only the least significant bits, but also the most significant bits. The likelihood of a bit being changed is related to its influence on the image quality. The LEB codes make it possible, thanks to their block structure, to manage the different types of bits at the same time. This scheme was ameliorated  in \cite{DSa12}.  In \cite{Bel19},  an algebraic study of cyclic LEB codes and some related results are discussed. In \cite{Bela19}, the authors give the existence conditions of infinite families of perfect LEB codes, and extended the notions of Hamming and simplex codes to linear error-block codes. In \cite{bel21a}, the tensor product of LEB codes was investigated. In \cite{bel22}, the first description of polynomial enumerator weight for LEB codes is hold. The author gives a simplest form of the polynomial weight enumerator for some particular families of LEB codes.

Let $ q $ be a prime power, $ \mathbb{F}_q$ is a finite field with $ q $ elements. Let $ n $ be a positive integer, the space $ \mathbb{F}_q^n $ can be viewed as a direct sum of spaces $ \mathbb{F}_q^{n_i} $ where $ i=1,2,\ldots,s $ are non null positive integers satisfying $ n=\sum_{i=1}^{s}n_i $ and $ n_1\geq n_2\geq\ldots \geq n_s\geq 1 $. Each vector in $ \mathbb{F}_q^n $ can be considered as a concatenation of $ s $ blocks $ u=(u_1,u_2,\ldots,u_s) $ where $ u_i\in \mathbb{F}_q^{n_i} $. Any change that happens inside a block causes a single error in the vector regardless to its magnitude. An LEB code is a linear subspace of $ \mathbb{F}_q^n $ endowed with the metric that measures the number of distinct blocks. This metric, clearly related to the integers $ n_i $ (lengths of blocks), is called the $ \pi $-metric where $ \pi $ is the partition of $ n $ noted $ \pi = [n_1 ][n_2 ]\ldots [n_s ] $ and called the type of the code. A classical linear error correcting code is a linear error-block code for which $ n_i=1 $ for $ i=1,\ldots,s $.

The family of cyclic codes has a significant role in the theory of error correcting codes. In the literature they have been extensively studied, since they have amazing algebraic properties. In fact, these codes are the most studied of all codes. Many well-known codes, such as BCH, Kerdock, Golay, Reed-Muller, Preparata, Justesen, and binary Hamming codes, are either cyclic codes or are constructed from cyclic codes. Quasicyclic codes are a generalization of cyclic codes.

In \cite{DSa12b} R. Dariti \textit{et al.} have considered a composition of an integer $ n=s\sum_{i=1}^{\lambda}n_i $ denoted by $ \pi=([n_1]\ldots[n_\lambda])^s$ where $ n_1> n_2>\ldots> n_s\geq 1 $ and defined a cyclic LEB code $ C $ over $ V_{\pi}=(V)^s $ where $ V=\mathbb{F}_q^{n_1}\oplus\ldots\oplus\mathbb{F}_q^{n_{\lambda}} $ and $ q $ is a prime power to be a code stable by the cyclic shift $ \sigma_{\pi} $ defined as follows:
$$  \begin{array}{cccc}
	\sigma_\pi: & \underbrace{V\oplus\ldots \oplus V}_{s \hspace{0.1cm} times}  & \longrightarrow &\underbrace{V \oplus\ldots\oplus V}_{s \hspace{0.1cm} times}  \\ 
	&(a_1,a_2,\ldots, a_s)  & \longmapsto  & (a_s,a_1,\ldots, a_{s-1})
\end{array} $$ 
\textit{i.e.} if $ c$ is a codeword of $ C $, then $$ \sigma_{\pi}(c)\in C. $$

The notation $ \pi=([n_1]\ldots[n_{\lambda}])^s$ means that each block of a codeword $ c $ in $ C $ is partitioned into $ j $ sub-blocks of length $ n_i\ (1\leq i\leq \lambda) $, and this block is of length $ n=s\sum_{i=1}^{\lambda}n_i $. Besides, the $ \pi- $weight of each codeword remain unchanged even if we change the partition used into the sub-block. Therefore, the partition into sub-blocks plays no role in the definition of cyclicity of LEB codes. That is why, in this paper we slightly modify this definition by talking about the cyclic LEB codes only when the blocks of codewords among an LEB cyclic code have the same length $ m $. i.e. the partition considered is $\pi = [m]^s\ (m ,s \in \mathbb{N}^{\star})$ where  $ n=m\times s $ is the length of the code.

The paper is organized as follows: in the next section, we give a servy of works that are held on the subject of post-quantum signatures. In Section $3$, we revisite the LEB cyclic codes construction. In Section $4$, We give a detailed algebraic construction of Quasi-cyclic LEB codes. In Section $5$, after having proving the NP-completness of decoding quasi-cyclci codes, we propose a decoding algorithm to this family of codes.  
Our scheme is presented in Section $6$, together with a detailed security analysis (Section $6$), and its performance are discussed in Section $7$. We conclude in Section $8$.

\section{Related work}
As conventional systems face threats from quantum computing since the proposed Shor algorithm \cite{shor} in 1994, post-quantum cryptography  renewed interest in code-based cryptography that started with McEliece \cite{mce} and Niederreiter \cite{nied}. Xinmei suggested a digital signature scheme \cite{xinmei} based on error-correcting codes in 1990, similar to  Rao-Nam \cite{nam} private-key cryptosystem, relying also on the difficult problem of factoring large matrices. Not long after that, Harn and Wang \cite{harn}  introduced a homomorphism attack on the Xinmei scheme in 1992, while proposing their own modified version, dealing with a nonlinear function, which was found vulnerable itself to chosen plaintext attacks by Alabbadi and Wicker, and to key derivation by Van Tilburg.\\
Due to Fiat-Shamir transformation \cite{fiat}, that is able to take Zero-Knowledge Identification Schemes (ZKID) into a signature, another approach was made possible by Stern \cite{stern} for code-based signature schemes.  However, these signature have high soundness errors, needing multiple repetitions which results in a longer output. Trapdoor constructions remain another option but with slow speed and large keys.\\
Courtois, Finiasz, and Sendrier (CFS) \cite{cfs} proposed another code-based signature with large keys using Niederreiter cryptosystem, 
as Kabastianskii, Krouk and Smeets (KKS) \cite{kks} have done  in 1997 with a signature relying on random error-correcting codes, instead of Goppa codes, but was later found unsecure.
\section{LEB cyclic codes}
In the literature, a cyclic code is a code that is invariant by a cyclic shift. That means that if the coordinates of a codeword are shifted to the next position, then the obtained word is a codeword. In the definition bellow, we define the cyclicity of a cyclic LEB code using the analogue method. 

Let $\pi$ be a partition of an integer $n$ of the form	$  \pi = [m]^s$.

Let $ V_{\pi}=\oplus^s_{i=1} \mathbb{F}_q^m$ . Each vector $u  \in \mathbb{F}_q^m$ can be written uniquely as $u=(u_1,u_2,..., u_s)$ with  $a_i\in V\ (i =1,...,s)$.

\begin{definition}\label{def1}
	An $[n, k, d]$ code $C$ of type $\pi = [m]^s$ is $\pi$-cyclic if for
	each $a  \in C$ we have $\sigma_\pi (a) \in C$ where		
	
	$$  \begin{array}{cccc}
		\sigma_\pi: & \underbrace{\mathbb{F}_q^m\oplus ... \oplus \mathbb{F}_q^m}_{s \hspace{0.1cm} times}  & \longrightarrow &\underbrace{\mathbb{F}_q^m \oplus ... \oplus\mathbb{F}_q^m}_{s \hspace{0.1cm} times}  \\ 
		&(u_1,u_2,..., u_s)  & \longrightarrow  & (u_s,u_1,..., u_{s-1})
	\end{array} $$
	$\sigma_\pi$ is a cyclic shift of one block to the right.
\end{definition}
\begin{remark}
	If $ \pi=[1]^n$, the classical definition of cyclic code is found by setting $ m=1 $.
\end{remark}
\begin{example}
	Let $ \pi=[2]^2 $. The LEB code $$C=\{000|000;010|110;110|010;100|100\} $$ of type $ \pi $ is a $ [6,2,2] $ $ \pi $-cyclic code.
\end{example}

\begin{definition}[The binary operation $ \star $]\label{def 2:11}
	Let $R_{\pi}= (\mathbb{F}_q[X],+,.) $ be the ring of polynomials with coefficients in $ \mathbb{F}_q $ and provided with the classical addition $ + $ and the classical multiplication $ . $ of polynomials in $ \mathbb{F}_q[X] $, and $ <X^n-1> $ is the principal ideal generated by the polynomial $ X^n-1 $ (i.e. $ <X^n-1>=\{f.(X^n-1)/f\in \mathbb{R}_q\} $). We define $ \star $ to be the mapping defined by:
	\begin{align*}
		\star : \frac{\mathbb{F}_q[X]}{<X^n-1>}\times\frac{\mathbb{F}_q[X]}{<X^n-1>}\longrightarrow&\frac{\mathbb{F}_q[X]}{<X^n-1>} \\
		(P(X),Q(X))\longmapsto & P(X)\star Q(X)\\
		&=X^{m-1}.P(X).Q(X).
	\end{align*}
	$ \star $ is a binary operation over $ \frac{\mathbb{F}_q^n [X]}{<X^n-1>} $ which verify the following properties:
\end{definition}
\begin{itemize}
	\item for  $ i$ and $ j$ in $\mathbb{N} $; $ X^{i}\star X^{j}= X^{i+m-1}.X^{j}=X^{j}\star X^{i}$ 	where $ . $ is the classical multiplication.
	\item for $i,j, k \in \mathbb{N}$;  $X^{i}\star(X^{j}\star X^{k})= (X^{i}\star X^{j})\star X^{k}$.
	\item  $ \overbrace{X\star\ldots\star X}^{\alpha\ time}=X^{\star\alpha}=X^{(\alpha-1)m+1} $ where $ \alpha $ is an integer $ \geq 1 $. 
	\item $ X^{\star 0}=\textbf{1}^{\star} $ the unity element of the binary operation $ \star $.
	\item $ 1\star X =X^m$ where $ 1 $ denotes here $ X^0 $.
	\item $ X^{\star s}\star P(X)=P(X) mod (X^n - 1)$. In fact, 
	\begin{align*}
		X^{\star s}\star P(X)&=X^{m-1}.X^{(s-1)m+1}.P(X)\\
		&=X^{ms}.P(X)\\
		&=P(X) mod (X^n - 1).
	\end{align*}
	\item $ 1\star X =X^m$ where $ 1 $ denotes here $ X^0 $
\end{itemize}
\begin{proposition}\label{prop 2:0}
	Let $\pi$ be a partition of a positive integer $n$ of the form	$  \pi = [m]^s$. Let $ \mathbb{F}_q $ be the finite field of $ q $ elements and $ R_{\pi}=\frac{\mathbb{F}_q[X]}{<X^n-1>}$ be the quotient ring. Then $  (R_{\pi},+,\star) $ where $ + $ is the classical addition of polynomial and $ \star $ is the polynomial multiplication defined by \ref{def 2:11}.	Then $ (R_{\pi},+,\star) $ is a commutative ring, with $\textbf{1}^{\star}=X^{n-m+1}$ is the unity element of the law $\star $.
\end{proposition}

\begin{proof}
	Let $ P(X) =\sum_{i=0}^{n-1}a_i\star X^i$, $ Q(X) =\sum_{i=0}^{n-1}b_iX^i$ and $ R(X)=\sum_{i=0}^{n-1}c_iX^i $ three polynomials in $ R_{\pi}$. Then 
	
	1. $ (R_{\pi},+) $ is an abelian group (where $ + $ is the usual polynomial addition).
	
	2. By construction,  the binary operation $\star $ is commutative and is distributive with respect to the usual addition $ + $.
	
	3.  The operation $ \star $ is associative. In fact
	
	\begin{align*}
		(P(X)\star Q(X))\star R(X) &=  (\sum_{i=0}^{n-1}a_i(X^i\star Q(X)))\star R(X) \\
		&=  (\sum_{i=0}^{n-1}a_iX^{i+m-1}.Q(X))\star R(X) \\
		&=  X^{m-1}.Q(X).(\sum_{i=0}^{n-1}a_i.X^i\star R(X)) \\
		&=  X^{m-1}.Q(X).(\sum_{i=0}^{n-1}a_{i}.X^{i+m-1}. R(X)) \\
		&=  X^{2m-2}.P(x).Q(X).R(X)                        
	\end{align*} 
	
	and 
	
	{\small \begin{align*}
			P(X)\star \left( Q(X)\star R(X) \right) =&  P(X)\star(\sum_{i=1}^{n-1}b_i(X^i\star R(X)))  \\
			=&  P(X)\star(\sum_{i=0}^{n-1}b_iX^{i+m-1}.R(X))\\
			=&  (P(X)\star\sum_{i=0}^{n-1}b_i.X^i).X^{m-1}.R(X) \\
			=&  (\sum_{i=0}^{n-1}b_i.(P(X)\star X^i)).X^{m-1}.R(X)  \\
			=&  (\sum_{i=0}^{n-1}b_i.(P(X). X^{i+m-1})).X^{m-1}.R(X)  \\
			=&  X^{2m-2}.P(x).Q(X).R(X).                         
		\end{align*}
	}
	4. The operation $ \star $ admits a unity element $ \textbf{1}^{\star}=X^{n-m+1} $. In fact $$\begin{array}{ccl}
		\textbf{1}^{\star}\star P(X)&=& X^{n-m+1}\star P(X) \\
		&=&X^{n-m+1+(m-1)}. P(X) \\
		&=&X^{n}. P(X)  \\
		&=& P(X).  \\
	\end{array}$$
	Since $ (R_{\pi},\star) $ is commutative, then $$P(X)\star \textbf{1}^{\star}= \textbf{1}^{\star}\star P(X) =P(X) $$
	
\end{proof}
Let $ C $ be an $ [n,k]_q $ $ \pi $-cyclic code of type $\pi=[m]^s$. Each codeword $ c=(c_1,c_2,\ldots,c_l)$ of $ C $ is associated with a unique polynomial 
$$C(X)=c_1(X)+X\star c_2(X)+\ldots + X^{\star (s-1)}\star c_s(X)$$ 
where for all $i=1,\ldots,s $, $ c_i(X)= \sum_{j=0}^{m-1}\alpha_{ij}.X^{j}$ and $ \alpha_{i,j} $ is the $ j^{th} $ component  of the vector $ c_i $.
%
\subsection{Algebraic representation}
Let $ C $ be an $ [n,k]_q $ $ \pi $-cyclic code of type $\pi=[m]^s$. Let $ c=(c_1,c_2,\ldots,c_s)$ be a codeword of $ C $. We associate with the vector $ c $ the polynomial 
$$C(X)=c_1(X)+X\star c_2(X)+\ldots + X^{\star (s-1)}\star c_s(X)$$ 
where $ c_i(X)= \sum_{j=1}^{m-1}\alpha_i.X^{j}$ for $i=1,\ldots,s $ and $ \alpha_j $ is the $ j^{th} $ component  of the vector $ c_i $.\\	
We have the following result:\\	
Let $ C $ be an $ [n,k]_q $ $ \pi $-cyclic code of type $\pi=[m]^s$ where $n=sm$. \\
\begin{theorem}
	A linear error-block code $ C $ is $ \pi $-cyclic if and only if $ C $ is an ideal of $( R_{\pi},+,\star) $.
\end{theorem}	
\begin{proof}
	
	Take $ c=(c_1,c_2,\ldots,c_s)$ a codeword of $ C $. Then  $$\hat{c} =\sigma_{\pi}(c)=(c_s,c_1,\ldots,c_{s-1}) $$ ($ \sigma_{\pi} $ is the cyclic shift defined in Definition (\ref{def1})) is also a codeword of $ C$
	and
	
\begin{multline*}
\hat{c}(X)= c_s(X)+X\star c_1(X)+\ldots +\\ X^{\star (s-1)}\star c_{s-1}(X)\\
       	  = c_s(X)+X\star c_1(X)+\ldots + X^{\star (s-1)}\star c_{s-1}(X)+\\ X^{\star s}\star c_{s}(X)-X^{\star s}\star c_{s}(X)\\
		  = X\star(c_1(X)+X\star c_2(X)+\ldots+X^{\star (s-1)}\star c_s(X))+\\ c_s(X)-X^{\star s}\star c_{s}(X)\\
		  =  X\star c(X)+c_s(X).(1-X^n)\\
		  = X\star c(X)\in \ R_{\pi}.
\end{multline*}

	Since $ \hat{c}(X)\in C $, then $ X\star c(X) \in C$. By definition, $ (C,+) $ is an abelian group, thus $ C $ is an ideal of $ R_{\pi} $. \\		
	Reciprocally, if $ C $ is an ideal of $ R_{\pi} $, then for all $ c(X) $ in $ C $,  $ r(X)\star c(X) $ is in $ R_{\pi} $ where $ r(X) $ is an arbitrary polynomial in $ R_{\pi} $. Let $ c(X)=\sum_{i=0}^{ms-1}\alpha_i X^i$ in $ C $ where $ \alpha\in\mathbb{F}_q $. Thus 
	\begin{multline*}
		c(X)= \sum_{i=0}^{m-1}\alpha_i X^i+\sum_{i=m}^{2m-1}\alpha_i X^i+\ldots+\\ \sum_{i=(s-1)m}^{n-1}\alpha_i X^i\\
		= \sum_{i=0}^{m-1}\alpha_i X^i+X^m.\sum_{i=m}^{2m-1}\alpha_i X^{i-m}+\ldots \\ 
		 + X^{m(s-1)}.\sum_{i=(s-1)m}^{n-1}\alpha_i X^{i-m(s-1)}\\
		= c_1(X)+X^m.c_2(X)+\ldots + X^{m(s-1)}.c_{s-1}(X)\\
		=c_1(X)+X\star c_2(X)+\ldots + X^{\star (s-1)}\star c_{s}(X).\\
	\end{multline*}
	where for all $ i=1,\ldots, s $: $$ c_i(X)=\sum_{j=(i-1)m}^{im-1} \alpha_{j-(i-1)m}X^{j-(i-1)m}$$. The polynomial $ c(X) $ is associated with vector $ (c_1,c_2,\ldots,c_s) $. Now, if we multiply $ c(X) $ by $ X $ we get $ X\star c(X)=X^m.c(X)=c_s(X)+X\star c_1(X)+\ldots + X^{\star (s-1)}\star c_{s-1}(X) $ which is a polynomial associated with the vector $ (c_l,c_1,\ldots,c_{s-1}) $. Thus, multiplying $ c(X) $ by $ X $ in $ R_{\pi} $ corresponds to a cyclic block shift. Finally, $ C $ is a $ \pi $-cyclic code.
\end{proof}

\begin{definition} 
	
	Let $ P(X) =\sum_{i=0}^{s-1}a_i(X)\star X^{\star i}$ be a polynomial of $ R_{\pi} $ with $ a_i(X)\in\mathbb{F}_q[X] $ such that $ a_{s-1}\neq 0 $. The integer $ s-1 $ is called the $ \star $-degree and is denoted by $\star-deg(P(X))=s-1$.
\end{definition}
\begin{definition}  Let $ n $ be a positive integer, and $ s $ and $ s' $ two integers $ \leq n $.
	Let $ P(X) =\sum_{i=0}^{s-1}a_i(X)\star X^{\star i}$, and $ Q(X) =\sum_{i=0}^{s'-1}b_i(X)\star X^{\star i}$ be two polynomials of $ R_{\pi} $, with $ a_i(X), b_i(X)\in\mathbb{F}_q[X] $ such that   $ a_{s-1}\neq 0  $ and  $ b_{s'-1}\neq 0 $. We say that $\star-deg(P(X))\leq \star-deg(Q(X)) $ if and only if ($ s<s' $ or ($ s=s' $ and $ deg(a_{s-1})<deg(b_{s'-1}) $)).
\end{definition}
\subsection{Polynomial Representation}
\label{sec 2:3}
\begin{theorem}\label{th 2:2}
	Let $ C $ be an $ [n,k]_q $ $ \pi $-cyclic code of type $ \pi=[m]^s $. 
	\begin{enumerate}
		\item There exist a unique monic $ g(X) $ in $ C $,  of minimal degree.
		\item $g(X)$ $ \star $-divides every word $c(X)$ in $ C $.
		\item $ g(X) $ $ \star $-divides $ X^{n}-1 $ in $ \mathbb{F}_q [X] $;
	\end{enumerate}
	The polynomial $g$ thus defined is called the generator polynomial of the code $C$. (Here the degree is related to the usual multiplication)
\end{theorem}

\begin{proof}\begin{enumerate}
		\item  Suppose that there exist two such unit polynomials $g_1(X)$ and $ g_2(X) $; then we can always construct a unitary polynomial of the form $a\star(g_1 - g_2)(X)$ ( for $ a $ a constant polynomial of $R_{\pi}$) which belongs to the code and which is of degree strictly less than $deg (g_1)$. But since $ deg (g_1) = deg (g_2) $ is the smallest possible degree, we get a contradiction and so $ g(X) $ is unique.
		\item  	Let $ c (X) $ be a non-zero codeword, then by  the Euclidean division (respecting the operation $ \star $) of $ c (X) $ by $ g (X) $, we obtain $ c (X) = g (X)\star q (X) + r (X)$  where $ deg (r (X)) <deg (g (X)) $. Then, 
		$ c (X) - g (X)\star q (X) = r (X)$. Since $ C $ is a $ \pi $-cyclic code and $ g (X) \in C $, then $ r (X) \in  C $. But $ g (X) $ is of minimal degree, so $ r(X)=0 $, Then ($2$) is done.
		\item Write $X^n-1 = g(X)\star q(X) + r(X)$ in $F[X]$, where $\deg r(X) < \deg g(X)$. In $ R_{\pi} $ we get  $g(X)\star q(X) = - r(X)\in C$, a contradiction (since $ g $ is of minimal degree), thus $ r(X)=0 $.
	\end{enumerate}
\end{proof}
\begin{theorem}\label{th 2:3}
	Let $ C $ be an $ [n,k]_q $ $ \pi $-cyclic code of type $ \pi=[m]^s $ and with generator polynomial $ g(X)=g_0(X)+X\star g_1(X)+\ldots+X^{\star r}\star g_r(X) $ where $ g_0,g_1,\ldots,g_r $ are polynomials in $ \frac{\mathbb{F}q[X]}{<X^n-1>}$.  Then 
	\begin{enumerate}
		\item  $ \dim{C}=k=s-r $.
		\item The matrix $G$ defined by   $ G=\left( \begin{array}{c}
			g(X)\\ 
			X\star g(X)\\ 
			\ldots\\
			X^{\star(s-1)}\star g(X)
		\end{array}\right)=$  
		{\scriptsize $ \left( \begin{array}{cccccccc}
				g_0(X)&g_1(X)&\ldots&g_r(X)&0&\ldots&0&0\\ 
				0&g_0(X)&g_1(X)&\ldots&g_r(X)&0&\ldots&0\\ 
				\ldots\\
				0&0&\ldots&0&g_0(X)&g_1(X)&\ldots&g_r(X)\\ 
			\end{array} \right) $} is a generator matrix of $C$.
	\end{enumerate}
\end{theorem}

\begin{proof}
	Let $ C $ be an $ [n,k]_q $ $ \pi $-cyclic code of type $ \pi=[m]^s $ and with generator polynomial $ g(X)=g_0(X)+X\star g_1(X)+\ldots+X^{\star r}\star g_r(X) $ where $ g_0,g_1,\ldots,g_r $ are polynomials in $ \frac{\mathbb{F}q[X]}{<X^m-1>}$, and let $G$ be the matrix defined in \ref{th 2:3}.
	
	Thanks to \ref{th 2:2}, $g(X)$ $\star$-divides $X^n-1$ in $\mathbb{F}_q[X]$, which ensures the existence of a monic $ h (X) $ such that $ g(X)\star h(X) = X^n - 1 $. Hence $ g_0 \neq 0 $. Besides, the rows of $ G $ are: $ g(X),X\star g (X),\ldots, X^{\star k-1}\star g(X)$ and $g(X)$ being the associated polynomial of an arbitrary codeword of $C$, then, the rows of $ G $ are linearly independent. Furthermore, by Theorem \ref{th 2:2} if $c(X) $ is the associated polynomial to a codeword $c$ of $C$, then $ c(X)=q(X)\star g (X) $ where $ q(X) $ is a polynomial such that $ deg (q (X)) <n-r $ (In fact $ deg (c (X)) <n $). Hence $ q(X) $ is of the form $$q(X)= q_0(X) + q_1(X)\star X+\ldots+ q_{s-r}(X)\star X^{\star (s-r-1)}$$ and 
	$$\begin{array}{ccc}
		c(X)&=& g(X)\star q_0(X) + g(X)\star q_1(X)\star X+\ldots+\\
		& & g(X)\star q_{s-r}(X)\star X^{\star( s-r-1)}.
	\end{array}$$
	Thus $c(X)$ is a linear combination of the rows of $ G $.
	Hence $ G $ is a generator matrix of $ C $.
\end{proof}
\begin{definition}\label{def 2:2}
	Let $ G $ and $ H $ respectively be the generator and the parity check matrix of a $ \pi $-cyclic code of type $ \pi=[m]^s $, we define $ G\star H^t $ by $ G\star H^t=G'.H^t$ where $ G' $ is the shifted matrix of $ G $ by $ m $ columns. We have $ G\star H^t=0 $.
\end{definition}
In the literature, a codeword $c$ of a linear code $C$ is encoded as $c.G$ where $G$ is a generator matrix of $C$. Hereafter, we presenting an algorithm to encode $\pi$-cyclic codes.

Let $ C $ be an $ [n,k]_q $ $ \pi $-cyclic code of type $ \pi=[m]^s $ and with generator polynomial $ g (X) $ then $ r = deg (g (X)) = s - k $ and the rows of $ G $ are $ g(X),X\star g (X),\ldots, X^{\star k-1}\star g(X)$.

Let $ u = (u_1, u_2, \ldots, u_k)\in\mathbb{F}_q^k $, then $ u (X) = u_1 + u_2\star X + \ldots + u_k\star X^{\star k-1} $ is the associated polynomial of $u$.

we have 
\begin{align*}
	u\star G &= u_1\star g(X) + u_2\star X \star g(X)+ \ldots + u_k\star X^{\star k-1}\star g(X)\\
	& = u (X)\star g (X).
\end{align*}

Since, $C$ is $\pi$-cyclic, then $u (X)\star g (X)\in C$. Thus, the polynomial message $u(x)$ is encoded as $ u (X)\star g (X) $.
\section{Quasi-cyclic LEB codes}\label{sect:3}

\subsection{Definition and Properties}
\begin{definition}\label{def 3:3}
	An LEB code $C$ of lenght $n=m\times s$ (where $s=p\times l$, $p$ and $l$ are positive integers greater than $1$.) is called an quasi-cyclic LEB (QCLEB) code of type $\pi=[m]^s$ and order $p$ iff any codeword which is cyclically block-right-shifted by $p$ blocks is again a codeword in $C$, which means that for all $c\in C, \ \sigma^p(c)\in C$. Such a code is represented by a parity check matrix $H$ consisting of $p\times m.p$ block-circulant matrices.
\end{definition}

\begin{result}\label{res:1}
	The product $G\star H^t$, will be defined this case by $ G\star H^t=G'.H^t$ where $ G' $ is the shifted matrix of $ G $ by $ mp $ columns. We have $ G\star H^t=0 $.
\end{result}

\begin{definition}
	Let $m$ and $p$  be two positive integers, and let $M$ be an $p\times mp$ matrix, $M$ is called a block-circulant matrix of order $p$ iff 
	$$ M=\begin{pmatrix}
		M_1&M_2&\ldots&M_p\\
		M_p&M_1&\ldots&M_{p-1}\\
		\vdots&\vdots&\vdots&\vdots\\
		M_2&\ldots&M_p&M_1
	\end{pmatrix}.$$
	
	where for all $i=1,\ldots,p$, $M_i\in\mathbb{F}_q^m$ is called a block of $mp$ elements.
\end{definition}
Let $M$ be a block-circulant matrix of order $p$, we can uniquely associate with $M$ a polynomial $M(X)$ of coefficient given by entries of the first rows of $M$. If $r=(M_1,M_2,\ldots,M_p)$ is the first rows of $M$, then the associate polynomial of $r$ is $$M(X)= \sum_{i=1}^{i=p} M_i(X)\star X^{\star i},$$ 
where $M_i(X)= \sum_{i=1}^{i=p}c_{i,j}.X^{j},$ for all $i=1,\ldots,p$, and $c_ij\in\mathbb{F}_q$. 
Adding or multiplying two  circulant-block matrices is equivalent to adding of multiplaying their their associated polynomials modulo $X^{\star p}-1$.
\begin{definition}
	The set of block-circulant matrices of order $p$ forms a ring isomorphic to the ring $\frac{\mathbb{F}_q[X]}{X^{\star p}-1}$.
\end{definition}
\begin{definition}[$\pi$-syndrome]
	Let $C$ be an $[n,k]_q$ LEB quasicode of type $\pi$, and let $H$ be a parity-check matrix of $C$. for any vector$e\in\mathbb{F}_q^{n-k}$, we define the $\pi$-syndrome of $y$ to be the vector $s_{\pi}=H\star y^{t}$, where $y^{t}$ is the transpose of the vector $y$.
\end{definition}
For any $s_{\pi} \in \mathbb{F}_q^{n-k}$ and parity-check matrix $H$, the set of vectors of $\mathbb{F}_q^n$ with syndrome $s_{\pi}$ is denoted by $S_H^{-1}(s_{\pi})=\left\{y \in \mathbb{F}_q^n: s=H\star y^T\right\}$. By definition, $S_H^{-1}(0)=\mathcal{C}$ for any parity-check matrix $H$ of $\mathcal{C}$. The vector space $\mathbb{F}_q^n / \mathcal{C}$ consists of all cosets $a+\mathcal{C}=\{a+c: c \in \mathcal{C}\}$ with $a \in \mathbb{F}_q^n$. There are exactly $q^{n-k}$ different cosets, each coset contain $q^k$ vectors, and form a partition of $\mathbb{F}_q^n$.

Finding a vector $y \in \mathbb{F}_q^n$ and its syndrome $s_{\pi}=H\star y^T$, decoding consists in finding a codeword $c \in \mathcal{C}$ closest to $y$ for the $\pi$-distance $(d_{\pi}(y, c) \leq t)$ or finding an error vector $e \in y+\mathcal{C}$ such that $w_{\pi}(e) \leq t$. In terms of algorithmic complexity, the corresponding decision problems are as follows:

\begin{definition}[Decisional Syndrome Decoding Problem (D-SDP)] Given a random matrix $H$, a syndrome $s_{\pi}$ and an integer $t>0$, determine if exists a vector $e$, with $w_{\pi}(e) \leq t$, such that $s_{\pi}=H\star e^T$.
\end{definition} 

\begin{definition}[Decisional Codeword Finding Problem (D-CFP)] Given a random matrix $H$ and an integer $t>0$, determine if exists a vector $e$, with $w(e) \leq t$, such that $H\star e^T=0$.
\end{definition}

\begin{definition}[quasi-cyclic LEB Syndrome Decoding Problem (QCLEBSDP)] Given a parity-check matrix $H$ of an QCLEB code, a syndrome $s_{\pi}$ and an integer $t>0$, find a vector $e$, with $w_{\pi}(e) \leq t$, such that $s=H\star e^T$.
\end{definition}

\begin{example}
	Let $C$ be the quasi-cyclic LEB code with parity-check matrix $H$ defined by 
	
	$$H=\left(\begin{array}{c|c}
		H_{1,1}  &  H_{1,2}\\\hline
		H_{2,1}  &  H_{2,2}\\\hline
		H_{3,1}  &  H_{3,2}\\\hline
		H_{4,1}  &  H_{4,2}
	\end{array}\right),$$
	
	where 
	$$ H_{1,1}=H_{2,2}=\left(\begin{array}{c|c}
		010  &  000\\\hline
		000  &  010
	\end{array}\right), $$
	
	$$H_{1,2}=H_{2,1}=\left(\begin{array}{c|c}
		000  &  000\\\hline
		000  &  000
	\end{array}\right),$$
	
	$$H_{3,1}=H_{3,2}=\left(\begin{array}{c|c}
		100  &  000\\\hline
		000  &  100
	\end{array}\right),$$ 
	
	$$H_{4,1}=\left(\begin{array}{c|c}
		000  &  001\\\hline
		001  &  000
	\end{array}\right),$$ 
	and
	$$H_{4,2}=\left(\begin{array}{c|c}
		000  &  101\\\hline
		101  &  000
	\end{array}\right),$$ 
	Using the result \ref{res:1}, we get 
	$$G=\left(\begin{array}{c|c|c|c}
		101 & 000 & 100 & 000  \\\hline
		100 & 000 & 101 & 000 
	\end{array}\right)$$ a generator matrix of $C$.
\end{example}
\begin{result}
	\label{res : 3}
	Let $C$ be an QCLEB code of type $\pi=[m]^{s} $ with length $n=m\times s$ (where $s=p\times l$, $p$ and $l$ are positive integers greater than $1$.) and dimension $k=p\cdot\left(n_0-r_0\right)$ ($n_0=l.m$), take a codeword $c=(c_1,\ldots,c_{l})$ of $C$, where $c_i\in\mathbb{F}_q^{mp}$ (for all $i=1,\ldots,l$). Then, the $\pi$-weight of $c$ is defined by the following: $$w_{\pi}=\sum_{i=1}^{l}w_{\pi}(c_i).$$
\end{result}

\subsection{Polynomial Representation}
\begin{theorem}
	\label{th 2:2:1}
	Let $C$ be an quasi-cyclic LEB code of type $\pi=[m]^s$ (where $s=l\times p\geq 1$) and index $p$ ($p$ an integer greater than $1$). Then similarly to LEB cyclic codes, $C$ is generated by a polynomial $g$ of the form $g(X)= \sum_{i=0}^{l-1} X^{\star pi}\star g_i(X)$ where $g_i(X)= \sum_{j=0}^{p-1}g_{ij}(X)\star X^{\star i}$, $g_{ij}(X)=\sum_{j=0}^{m-1}\alpha_jX^j$ and $\alpha_j\mathbb{F}_q$.
\end{theorem}
\begin{proof}
    Take $ c=(c_1,c_2,\ldots,c_l)$ a codeword of $ C $. Then  $\hat{c} =\sigma_{\pi}(c)=(c_l,c_1,\ldots,c_{l-1}) $ ($ \sigma_{\pi} $ is the QC-LEB shift of the codeword $c$ and is also a codeword of $ C$
and $$\begin{array}{ccl}
\hat{c}(X)&=& c_l(X)+X\star c_1(X)+\ldots + X^{\star (l-1)}\star c_{l-1}(X)\\
&=& c_l(X)+X\star c_1(X)+\ldots + X^{\star (l-1)}\star c_{l-1}(X)+ \\
& & X^{\star l}\star c_{l}(X)-X^{\star l}\star c_{l}(X)\\
&=& X\star(c_1(X)+X\star c_2(X)+\ldots + \\
& & X^{\star (l-1)}\star c_l(X))+c_l(X)-X^{\star l}\star c_{l}(X)\\
&=&  X\star c(X)+c_l(X).(1-X^n)\\
&=& X\star c(X)\in \  R_{\pi}.
\end{array}$$
Since $ \hat{c}(X)\in C $, then $ X\star c(X) \in C$. By definition, $ (C,+) $ is an abelian group, thus $ C $ is an ideal of $  R_{\pi} $.
Reciprocally, if $ C $ is an ideal of $  R_{\pi} $, then for all $ c(X) $ in $ C $,  $ r(X)\star c(X) $ is in $ C $ where $ r(X) $ is an arbitrary polynomial in $  R_{\pi} $. Let $ c(X)=\sum_{i=0}^{ml-1}\alpha_i X^i$ in $ C $ where $ \alpha_i\in\mathbb{F}_q $. Thus
$$\begin{array}{ccl}
c(X)&=& \sum_{i=0}^{m-1}\alpha_i X^i+\sum_{i=m}^{2m-1}\alpha_i X^i+\ldots+\\
& & \sum_{i=(l-1)m}^{n-1}\alpha_i X^i\\
&=& \sum_{i=0}^{m-1}\alpha_i X^i+X^m.\sum_{i=m}^{2m-1}\alpha_i X^{i-m}+\ldots \\ 
& & +X^{m(l-1)}.\sum_{i=(l-1)m}^{n-1}\alpha_i X^{i-m(l-1)}\\
&=& c_1(X)+X^m.c_2(X)+\ldots+X^{m(l-1)}.c_{l-1}(X) \\
& & \mbox{ by definition of the low } \star \mbox{ we get }\\
&=&c_1(X)+X\star c_2(X)+\ldots + X^{\star (l-1)}\star c_{l}(X).\\
\end{array}$$
where for all $ i=1,\ldots, l $: $$ c_i(X)=\sum_{j=(i-1)m}^{im-1} \alpha_{j-(i-1)m}X^{j-(i-1)m}.$$ The polynomial $ c(X) $ is associated with vector $ (c_1,c_2,\ldots,c_l) $. Now, if we multiply $ c(X) $ by $ X $ we get $$ X\star c(X)=c_l(X)+X\star c_1(X)+\ldots + X^{\star (l-1)}\star c_{l-1}(X) $$ which is a polynomial associated with the vector $$ (c_l,c_1,\ldots,c_{l-1}). $$ Thus, multiplying $ c(X) $ by $ X $ in $  R_{\pi} $ corresponds to a QC-LEB block shift. Finally, $ C $ is a QC-LEB code.\\
Let $ c (X) $ be a non-zero word of the code, then if we make the Euclidean division (respecting the law $ \star $) of $ c (X) $ by a unitary polynomial $ g (X) $, we obtain $ c (X) = g (X)\star q (X) + r (X)$  with $ deg (r (X)) <deg (g (X)) $. \\
Then,
$ c (X) - g (X)\star q (X) = r (X)$. Since $ C $ is a $ \pi $-cyclic and $ g (X) \in C $, then $ r (X) \in  C $. But as $ g (X) $ is of minimal degree, so $ r(X)=0 $, which proves the result.\\
Let us write $X^n-1 = g(X)\star q(X) + r(X)$ in $F[X]$, where $\deg r(X) < \deg g(X)$. In $ R_{\pi} $ this says $g(X)\star q(X) = - r(X)\in C$, a contradiction (since $ g $ is of minimal degree) unless $ r(X)=0 $. Thus $ g(X) $ $ \star $-divides $ X^{n}-1 $ in $ \mathbb{F}_q [X] $.\\
Let $c(X)\in C $, Then $ c(X)=q(X)\star g (X) $ where $ q(X) $ is a polynomial such that $ deg (q (X)) <n-r $ since $ deg (c (X)) <n $. Hence $ q(X) $ is of the form $$q(X)= q_0(X) + q_1(X)\star X+\ldots+ q_{l-r}(X)\star X^{\star (l-r-1)}$$ and
$$\begin{array}{ccl}
    c(X) &=& g(X)\star q_0(X) + g(X)\star q_1(X)\star X + \\
         &=& \ldots+ g(X)\star q_{l-r}(X)\star X^{\star( l-r-1)}
\end{array}$$
and $c(X)$ is a linear combination of $g_i(X)$. 
\end{proof}
\begin{result}
	\label{res : 4}
	An QC-LEB code $C$ of type $\pi=[m]^{s} $ with length $n=m\times s$ (where $s=p\times l$, $p$ and $l$ are positive integers greater than $1$.) is generated by a matrix $G$ of the form $$G=\left(\begin{array}{c}
		g(X) \\
		X^{\star p\times 1}\star g(X) \\
		\vdots\\
		X^{\star p\times (l-1)}\star g(X)=X^{\star(s-p)}\star g(X)
	\end{array}\right)$$
	In fact, for all $c(X)\in C$ we have $X^{\star p\times i}\star g(X)\in C$ where $i$ is a positive integer greater than $1$.
\end{result}

\begin{example}
	\label{exa 2:1}
	Take the LEB code generated by the matrix 
	$$G=\left(\begin{array}{c|c|c|c}
		101 & 000 & 100 & 000  \\\hline
		100 & 000 & 101 & 000 
	\end{array}\right).$$
	
	$C$ is an QCLEB code of type $\pi=[3]^4$ of index $2$ and dimension $k=2$ and lenght $n=12$.
	A generator polynomial of $C$ is $$g(X)= g_0(X)\star\textbf{1}^{\star}+   g_1(X)\star X^{\star 2}$$ where $g_0(X)= X^2\star X^{\star 1}$ and $g_1(X)= (X^2+1)\star X^{\star 1}$.
	
	Thus, 
	\begin{align*}
		X^{\star 2}\star g(X)&= g_0(X)\star X^{\star 2}+   g_1(X)\star X^{\star 4}\\
		&= g_0(X)\star X^{\star 2}+   g_1(X)\star \textbf{1}^{\star}
	\end{align*}
	
	and since $X^{\star 4}\star g(X)=g(X)$, then $$G=\left(\begin{array}{c}
		g(X)\\
		X^{\star 2}\star g(X)
	\end{array}\right)$$ is a generator matrix of $C$.
\end{example}

\begin{definition}
	Let $ v(X) \in\mathbb{F}_q[X]$ and $ g(X) $ is a polynomial generator of  An QCLEB code $C$ of type $\pi=[m]^{s} $ with length $n=m\times s$ (where $s=p\times l$, $p$ and $l$ are positive integers greater than $1$.), we define  $ R_{g(X)}(v(X)) $ to be the unique polynomial $ r(X) $ verifying \begin{equation}
		v(X)=g(X)\star f(X) +r(X). \label{equa 2:111}
	\end{equation}  with $ r(X)=0 $ or $ \star-deg(r(X))\leq \star-deg(g(X)) $. (The existence and the uniqueness of $ r(X) $ verifying \ref{equa 2:111} are due to the existence and the uniqueness of $ r(X) $ such that $ v(X)=g(X).X^{m-1}.f(X) +r(X)$, and  $ r(X)=0 $ or $ deg(r(X))< deg(g(X)) $).
\end{definition}

In the following properties satisfied by the function $ R_{g(X)} $.
\begin{theorem}\label{th 2:5}
	With the preceding notation the  following hold:
	\begin{enumerate}
		\item $ R_{g(X)}(av(X) + bv'(X)) = a R_{g(X)}(v(X)) + bR_{g(X)}(v'(X))$ for all $ v(X), v'(X)\in\mathbb{F}_q[X]$ and all $a, b\in\mathbb{F}_q[X].$
		\item   $ R_{g(X)}(v(X)+a(X)\star(X^n-1)) = R_{g(X)}(v(X)) $.
		\item  $ R_{g(X)}(v(X))=0 $ iff $ v(X) mod (X^{\star s}-\textbf{1}^{\star}) \in C. $
		\item  If $ c(X) \in C, $ then $ R_{g(X)}(c(X) + e(X)) = R_{g(X)}(e(X)) $.
		\item $R_{g(X)}(v(X)))=v(X)$ lf $ \star-deg(v(X))< \star-deg(g(X)) $.
	\end{enumerate}
\end{theorem}
\begin{proof}
	\begin{enumerate}
		\item Let $ v(X), v'(X)\in R_{\pi}$ and  $a, b\in\mathbb{F}_q[X]$, and set $r_1(X)=R_{g(X)}(av(X))$ and $r_2(X)=R_{g(X)}(bv'(X))$. Then there are two polynomials $ f_1(X) $ and $ f_2(X) $ in $ R_{\pi} $ such that $ av(X)=f_1(X)\star g(X)+r_1(X) $ and $ bv'(X)=f_2(X)\star g(X)+r_2(X) $ and $ \star-deg(r_i(X))<  \star-deg(g(X))$ or $ r_i(X)=0 $ for $ i=1,2 $. Therefore, $ av(X)+bv'(X)=(f_1(X)+f_2(X))\star g(X)+(r_1+r_2)(X)$. Since 
		
		$ \star-deg(r_1(X)+r_2(X))\leq max(r_1(X),r_2(X))\leq l-k, $
		then \begin{align*}
			R_{g(X)}(av(X) + bv'(X))&=(r_1+r_2)(X)\\
			&=a R_{g(X)}(v(X)) \\
            &+bR_{g(X)}(v'(X)).
		\end{align*} 
		\item According to the result above, $ R_{g(X)}(v(X)+a(X)\star(X^n-1)) = R_{g(X)}(v(X)) + R_{g(X)}(a(X)\star(X^n-1))$. Since $ g(X) $ $ \star $-divides $X^{\star s}-\textbf{1}^{\star}$, then  $ R_{g(X)}(a(X)\star(X^{\star s}-\textbf{1}^{\star}))=0$, and $ R_{g(X)}(v(X)+a(X)\star(X^n-1)) = R_{g(X)}(v(X))$
		\item We have \begin{align*}
			v(X) mod (X^{\star s}-\textbf{1}^{\star}) \in C\Leftrightarrow & v(X)\ \star-divides\ g(X)\\
			&\Leftrightarrow R_{g(X)}(v(X))=0
		\end{align*} 
		\item Let $ c(X) \in C$, then $ R_{g(X)}(c(X))=0$ and then
		\begin{align*}
			R_{g(X)}(c(X) + e(X)) &= R_{g(X)}(c(X)) +R_{g(X)}(e(X))\\
			&=R_{g(X)}(e(X))
		\end{align*}
		\item lf $ \star-deg(v(X))<  \star-deg(g(X)) $ then $ v(X)=0\star g(X)+v(X) $, and $R_{g(X)}(v(X)))=v(X)$.
	\end{enumerate}
\end{proof}
\begin{definition}Let $ C $ be an QCLEB code of type $\pi=[m]^{s} $ with length $n=m\times s$ (where $s=p\times l$, $p$ and $l$ are positive integers greater than $1$.).
	We define the syndrome polynomial $ S_{\pi}(v(X)) $ of any polynomial $ v(X)\in R_{\pi} $ to be
	$ S_{\pi}(v(X)) = R_{g(X)}(v(X)) $.	
\end{definition}
\begin{example}
	\label{exa 2:2}
	For the code defined in the example \ref{exa 2:1}. Suppose that we have sent the word $c(X)= (x^2\star X^{\star 1})\star \textbf{1}^{\star}+((X^2+1)\star X^{\star 1})\star X^{\star 2}$ and that we have receive the word $y(X)=(x^2\star X^{\star 1})\star \textbf{1}^{\star}+(X^2\star \textbf{1}^{\star}+ (X^2+1)\star X^{\star 1})\star X^{\star 2}$. Then the syndrome polynomial of $y(X)$ is :  $s_{\pi}(y(X))=R_{g(X)}=(X^2\star \textbf{1}^{\star})\star X^{\star 2}$
\end{example}
\section{Decoding  quasi-cyclic LEB Codes}

\subsection{NP-completeness of decoding QC LEB codes}
According to the section \ref{sect:3}, the parity-check matrix of an QCLEB code of type $\pi=[m]^{s} $ ($s=p\times l$ with $l$ and $p$ are two positive integer greater than $1$) with length $n=m\times s$ and dimension $k=p \cdot\left(n_0-r_0\right)$ (with $n_0=mp$) is seen as follows
$$
\left(\begin{array}{c|c|c|c}
	H_{1,1} & H_{1,2} & \ldots & H_{1,l} \\
	H_{2,1} & H_{2,2}& \ldots & H_{2,l} \\
	\vdots & \vdots & \ddots & \vdots \\		
	H_{r_0,1}&H_{r_0,2}&\ldots&H_{r_0,l}
\end{array}\right)
$$
where each matrix $H_{i,j}$ with $1 \leq i \leq r_0$ and $1 \leq j \leq l$, is a $r_0\times n_0$ block-circulant matrix, that is, a matrix of the form
$$
\left(\begin{array}{c|c|c|c}
	a_1 & a_2 & \ldots & a_{p} \\
	a_{p} & a_1 & \ldots & a_{p-1} \\
	\vdots & \vdots & \ddots & \vdots \\
	a_2 & a_3 & \ldots & a_1
\end{array}\right)
$$
with $a_i \in \mathbb{F}_q^m, 1 \leq i \leq p$.

A block-circulant matrix of size $p\times m.p$ is associated through an isomorphism to a polynomial in $x$ with coefficients over $\mathbb{F}_q^m$ given by the elements of the first row of the matrix:
$$
a(x)=\sum_{i=0}^{p-1} a_{i+1}\star x^{\star i}
$$
\begin{theorem}
	\label{thm 3:3} let $p$, $m$ and $l$ three positive integers greater than $1$. 
	Take any $l$ matrices each of size $r_0\times n_0$ where $n_0=p\times m$. The parity-check matrix 
	
	\begin{equation}\label{equa 4:0}
		H=\left(\begin{array}{c|c|c|c}
			H_{1} & H_{2} & \ldots & H_{l} \\
			H_{2} & H_{3}& \ldots & H_{p-2} \\
			\vdots & \vdots & \ddots & \vdots \\		
			H_{2}&H_{3}&\ldots&H_{1}
		\end{array}\right)
	\end{equation}	
	
	defines a QCLEB code with lenght $n=p\times n_0$ and dimension $k=p.(n_0-r_0)$.
\end{theorem}

\begin{proof}
	Obviously,  $$H= 
	\left(\begin{array}{c|c|c|c}
		H_{1} & H_{2} & \ldots & H_{l} \\
		H_{2} & H_{3}& \ldots & H_{p-2} \\
		\vdots & \vdots & \ddots & \vdots \\		
		H_{2}&H_{3}&\ldots&H_{1}
	\end{array}\right)
	$$
	is an $r\times n$ matrix.

	Given a generic codeword $c$ of length $n$, it can be seen as 
	$c=\left(c_{1}, c_{2}, \ldots, c_{l}\right)$ where each $c_{i}, 1 \leq i \leq l$ has size $n_{0}$. By definition of an QCLEB code of index $p$, we have 
	$H\star c^{t}=0$, which is equivalent to $H.(c')^{t} =0$ where $c'$ is the shifted vector of $ G $ by $ n_0=mp $ positions,   that is:
	$$\begin{gathered}
		H_{1}.c_{l}^{t}+H_{2}.c_{1}^{t}+ \cdots +H_{l}.c_{l-1}^{t}=0 \\
		H_{l} \cdot c_{l}^{t}+H_{1} \cdot c_{1}^{t}+\cdots+H_{l-1} \cdot c_{l-1}^{t}=0 \\
		\vdots \\
		H_{2} \cdot c_{l}^{t}+H_{3} \cdot c_{1}^{t}+\cdots+H_{1} \cdot c_{l-1}^{t}=0
	\end{gathered}$$
	
	Denoting by $c^{(x)}, 1 \leq x \leq l-1$, the right cyclic shift of $c$ in $x \cdot n_{0}$ positions and from the definition of QCLEB code, we have that $c^{(x)}$ is a codeword too. Therefore, the equations
	
	\begin{equation}
		H\star (c^{(x)})^{t}=0, \  1 \leq x \leq l-1. \label{equa 4:1}
	\end{equation}
	
	must be satisfied. The equations $H \star\left(c^{(1)}\right)^{t}$ are
	
	$$\begin{gathered}
		H_{1}.c_{l-1}^{t}+H_{2}.c_{l}^{t}+H_{3}.c_{1}^{t}+ \cdots +H_{l}.c_{l-2}^{t}=0 \\
		H_{l}.c_{l-1}^{t}+H_{1}.c_{l}^{t}+H_{2}.c_{1}^{t}+ \cdots +H_{l}.c_{l-2}^{t}=0 \\
		\vdots \\
		H_{2}.c_{l-1}^{t}+H_{3}.c_{l}^{t}+H_{4}.c_{1}^{t}+ \cdots +H_{1}.c_{l-2}^{t}=0 
	\end{gathered}$$
	
	where the following pattern is observed:
	
	\begin{itemize}
		\item Equation 1 of $H \star\left(c^{(1)}\right)^{t}$ is the equation $l$ of $H\star c^{t}=0$.
		\item Equation 2 of $H \star\left(c^{(1)}\right)^{t}$ is the equation 1 of $H\star c^{t}=0$, and so on.
	\end{itemize}
	
	That is, all the equations of $H \star\left(c^{(1)}\right)^{t}$ coincide with an equation of $H\star c^{t}=0$, therefore, it is satisfied that $H \star\left(c^{(1)}\right)^{t}=0$. Let us now consider the equations $H \star\left(c^{(2)}\right)^{t}$
	
	$$
	\begin{gathered}
		H_{1} \cdot c_{l-2}^{t}+H_{2} \cdot c_{l-1}^{t}+\cdots+H_{l} \cdot c_{l-3}^{t} \\
		H_{l} \cdot c_{l-2}^{t}+H_{1} \cdot c_{l-1}^{t}+\cdots+H_{l-1} \cdot c_{l-3}^{t} \\
		\vdots \\
		H_{2} \cdot c_{l-2}^{t}+H_{3} \cdot c_{l-1}^{t}+\cdots+H_{1} \cdot c_{l-3}^{t}
	\end{gathered}
	$$
	
	Similarly, the following pattern is observed:
	
	\begin{itemize}
		\item Equation 1 of $H \star\left(c^{(2)}\right)^{t}$ is the equation $l$ of $H \star\left(c^{(1)}\right)^{t}=0$.
		\item Equation 2 of $H \star\left(c^{(2)}\right)^{t}$ is the equation 1 of $H \star\left(c^{(1)}\right)^{t}=0$, and so on.
	\end{itemize}
	
	All the equations of $H \star\left(c^{(2)}\right)^{t}$ coincide with an equation of $H \star\left(c^{(1)}\right)^{t}=$ 0 , therefore $H \star\left(c^{(2)}\right)^{t}=0$. Applying the same reasoning for all $x, 1 \leq x \leq l-1$ in $H \star\left(c^{(x)}\right)^{t}$, it can be seen that the general pattern is as follows:
	
	\begin{itemize}
		\item Equation 1 of $H \star\left(c^{(x)}\right)^{t}$ is the equation $l$ of $H \star\left(c^{(x-1)}\right)^{t}$.
		\item Equation $i, 2 \leq i \leq l$, of $H \star\left(c^{(x)}\right)^{t}$ is the equation $i-1$ of $H \star\left(c^{(x-1)}\right)^{t}$.
	\end{itemize}
	
	Therefore, the equations \ref{equa 4:1} are satisfied and the code is QCLEB.
\end{proof}
The converse of the previous theorem also holds.

\begin{theorem}
	If $C$ is an QCLEB code of length $n=l \cdot n_{0}$ and dimension $k=l \cdot\left(n_{0}-r_{0}\right)$, then the matrix
	$$H=\left(\begin{array}{cccc}
		H_{1} & H_{2} & \ldots & H_{l} \\
		H_{l} & H_{1} & \ldots & H_{l-1} \\
		\vdots & \vdots & \ddots & \vdots \\
		H_{2} & H_{3} & \cdots & H_{1}
	\end{array}\right)$$
	is a parity-check matrix for $C$, where each matrix $H_{i}, 1 \leq i \leq l$ has size $r_{0} \times n_{0}$.  
\end{theorem}
\begin{proof}
	Let $H^{\prime}$ be a parity-check matrix of a QC code of length $n=p \cdot n_{0}$ and dimension $k=p \cdot\left(n_{0}-r_{0}\right)$. Then $H^{\prime}$ can be written as follows
	
	$$
	H^{\prime}=\left(\begin{array}{cccc}
		H_{1,1} & H_{1,2} & \ldots & H_{1,l} \\
		H_{2,1} & H_{2,2} & \ldots & H_{2,l} \\
		\vdots & \vdots & \ddots & \vdots \\
		H_{l,1} & H_{l,2} & \ldots & H_{l,l}
	\end{array}\right)
	$$
	
	where each matrix $H_{i j}, 1 \leq i, j \leq l$ has size $r_{0} \times n_{0}$. Since $C$ is an QCLEB code, for any codeword $c=\left(c_{1}, c_{2}, \ldots, c_{l}\right)$, the equalities
	
	$$
	H^{\prime} \star\left(c^{(x)}\right)^{t}=0, 1 \leq x \leq l-1
	$$
	
	are satisfied. 
	
	Note that the equation \ref{equa 4:0} in $H^{\prime} \star c^{t}=0$ is the same as \ref{equa 4:1} in $H^{\prime} \star\left(c^{(1)}\right)^{t}=0$. Then is satisfied:
	
	$$
	\begin{gathered}
		H_{1,1}=H_{2,2} \\
		H_{1,2}=H_{2,3} \\
		H_{1,3}=H_{2,4} \\
		\vdots \\
		H_{1,l}=H_{2,1}
	\end{gathered}
	$$
	
	Similarly, we can consider the equation $3$ in $H^{\prime} \star\left(c^{(2)}\right)^{t}=0$ and equation $4$ in $H^{\prime} \star\left(c^{(3)}\right)^{t}=0$. This reasoning is extended to consider the equation $p$ in $H^{\prime} \star\left(c^{(l-1)}\right)^{t}=0$ and it is obtained that
	
	$$
	\begin{gathered}
		H_{1,1}=H_{2,2}=H_{3,3}=\cdots=H_{l,l)} \\
		H_{1,2}=H_{2,3}=H_{3,4}\cdots=H_{l,1} \\
		H_{1,3}=H_{2,4}==H_{3,5}\cdots=H_{l,2} \\
		\vdots \\
		H_{1,l}=H_{2,1}=H_{3,2}=\cdots=H_{l\times (l-1)}
	\end{gathered}
	$$
	The previous equalities then mean that $H^{\prime}$ has the form (1), that is, $H^{\prime}=H$.    
\end{proof}

From the previous theorems, it is not difficult to deduce the proof of the following result.
\begin{theorem}
	The $Q C$-SDP is $N P$-complete.
\end{theorem} 
\begin{proof}
	The general idea of the proof is the following. Starting from an instance of the SDP, an instance of the QCLEB-SDP is constructed in polynomial time and it is shown that if the latter is solved efficiently, then the instance of the former is necessarily solved efficiently.
	
	Let $(H,t,s_y)$ be an instance of the SDP and consider an QCLEB code with length $n=l\cdot n_{0}$ and dimension $k=l\cdot\left(n_{0}-r_{0}\right)$. The QCLEB-SDP instance $\left(H^{\prime}, t^{\prime}, s^{\prime}_y\right)$ is defined as follows. The matrix $H^{\prime}$ is
	
	$$
	H^{\prime}=\left(\begin{array}{cccc}
		H & H_{2} & \ldots & H_{l} \\
		H_{l} & H & \ldots & H_{l-1} \\
		\vdots & \vdots & \ddots & \vdots \\
		H_{2} & H_{3} & \ldots & H
	\end{array}\right)
	$$
	
	where the matrices $H_{i}, 1 \leq i \leq l$ are any matrices of size $r_{0} \times n_{0}$ randomly selected. Let $x \in \mathbf{F}_{q}^{n}$ be any vector. Because $n=l \cdot n_{0}$, $x$ can be seen as $x=\left(x_{1}, x_{2}, \ldots, x_{l}\right)$ where each $x_{i},\ 1 \leq i \leq l$ has length $n_{0}$. The value of $t^{\prime}$ is taken as $t+h$, with $h=\sum_{i=1}^{l} w_{\pi}\left(x_{i}\right)$. The syndrome $s^{\prime}_y$ is defined as $s^{\prime}_y=\left(s_y+s_{y,1}, s_{y,2}, \ldots, s_{y,l}\right)$ where
	
	$$
	\begin{gathered}
		H_{2} \cdot x_{1}^{T}+\cdots+H_{l} \cdot x_{l-1}^{T}=s_{y,1} \\
		H_{l} \cdot x_{l}^{T}+H \cdot x_{1}^{T}+\cdots+H_{l-1} \cdot x_{l-1}^{T}=s_{y,2} \\
		\vdots \\
		H_{2} \cdot x_{l}^{T}+H_{3} \cdot x_{1}^{T}+\cdots+H \cdot x_{l-1}^{T}=s_{y,l}
	\end{gathered}
	$$
	
	and each $s_{y,i}, 1 \leq i \leq l-1$ has length $r_{0}$.
	
	Let $\mathcal{A}$ be an algorithm capable of solving the instance $\left(H^{\prime}, t^{\prime}, s^{\prime}_y\right)$ of the QCLEB-SDP. This means that through $\mathcal{A}$, we can find a vector $e^{\prime}=\left(e_{1}, e_{2}, \ldots, e_{n}\right) \in \mathbb{F}_{q}^{n}$ with $w\left(e^{\prime}\right) \leq t^{\prime}$ such that
	
	$$
	H^{\prime} \star (e')^{T}=s^{\prime}
	$$
	
	The vector $e^{\prime}$ can be seen as $e^{\prime}=\left(e_{1}, e_{2}, \ldots, e_{l}\right)$ where each $e_{i}, 1 \leq$ $i \leq l$ has length $n_{0}$ and $\sum_{i=1}^{l} w_{\pi}\left(e_{i}\right)=h$. Thus, we have that $H^{\prime} \star (e')^{T}=s^{\prime}$ is
	
	$$
	\begin{gathered}
		H.x_l^T +H_{2} \cdot e_{1}^{T}+\cdots+H_{l} \cdot e_{l-1}^{T}=s_y + s_{y,1} \\
		H_{l} \cdot e_{l}^{T}+H \cdot e_{1}^{T}+\cdots+H_{l-1} \cdot e_{l-1}^{T}=s_{y,2} \\
		\vdots \\
		H_{2} \cdot e_{l}^{T}+H_{3} \cdot e_{1}^{T}+\cdots+H\cdot e_{l-1}^{T}=s_{y,l}
	\end{gathered}
	$$
	
	By definition, $s_{y,1}=\sum_{i=1}^{l-1} H_{i+1} \cdot e_{i}^{T}$. Substituting $s_{0}$ in the first of the above equations we get
	
	$$
	H \cdot e_{l}^{T}=s
	$$
	
	Since $w_{\pi}(e) \leq t^{\prime}$, we have that $w_{\pi}\left(e_{l}\right) \leq t^{\prime}-h=t$. Therefore, $e_{l}$ is a solution of the instance of the SDP. In this way, a solution of the QC-SDP allows to find in polynomial time, a solution of an instance of the SDP.
\end{proof}

\subsection{Meggitt-Like Decoding of QC-LEB  Codes}
Here we generalize a technique for decoding LEB QC codes called Meggitt decoding, this technique was introduced by J. E. Meggitt in $ 1960 $ \cite{Meggitt1,Meggitt2}.  Meggitt decoder is acceptable in practice in terms of substantially less complication and providing the best interference protection in comparison to a syndrome decoder. Another advantage of this method is that the decryption table, which contains all the error words and all the syndromes, is replaced by a table containing only the syndromes of the error words whose last symbol is incorrect. This saves a lot of memory space and time.\\
\\
The Meggitt decoder performs a symbol-by-symbol decoding. First, an erroneous component of the received w
encoding and decoding, also captures the cyclic nature and
reinforces orthogonality, a property that remains an essen-
tial component for error detection and correction. The study
illustrates the decoding of QC-LEB codes and defines the
syndrome polynomial while addressing the computational
complexity and emphasizing their NP-completeness. This
significant theoretical with practical implementation results
underscore the inherent difficulty of the decoding problem in
our case, that can be leveraged as a groundwork for other opti-
mised schemes where efficiency and security are paramount.
This was utilised to sketch a signature scheme with polyno-
mial operations over a finite field, all within a zero-knowledge
framework. Security analysis performed on this construction
proves uniform distribution, high entropy, and the absence of
exploitable patterns, which was affirmed by statistical tests.
Efficiency evaluations of the scheme exhibits small signature
sizes and short execution times for large inputs offering a solid
foundation for future cryptographic work aimed at further
optimizations and broader deployment.\\
\\
Take $C$ an quasi-cyclic LEB code of type $\pi=[m]^s$ (where $s=l\times p\geq 1$) and index $p$ ($p$ an integer greater than $1$). Now suppose that $C$ is generated by a polynomial $g$ of the form $g(X)= \sum_{i=0}^{l-1} X^{\star pi}\star g_i(X)$ where $g_i(X)= \sum_{j=0}^{p-1}g_{ij}(X)\star X^{\star i}$, $g_{ij}(X)=\sum_{j=0}^{m-1}\alpha_jX^j$ and $\alpha_j\mathbb{F}_q$, and then let us describe the Meggitt-like decoding algorithm and use an example to illustrate each step.
\begin{itemize}
	\item[$I$:] Let $c(X)$ be the word sent, $y(X)$ the word received, and $e(X)$ the error word with $w_{\pi}(e(X)) <t_{\pi}$. Find the syndrome table of errors $e(X)$ whose index component $n-1$ is erroneous.
	\begin{example}\label{exa 2:3}
		For the code $ C $ defined in the \ref{exa 2:2}, we calculate the error patterns syndromes. (cf. Table \ref{tab 1}).
		
		\begin{table} \label{tab 1}
			\centering
			$\begin{array}{|l|l|}
				\hline
				e(X)&S_{\pi}(e(X))\\\hline 
				((X^2)\star X^{\star 1})\star X^{\star 2}& ((X^2)\star X^{\star 1})\star X^{\star 2}\\\hline
			\end{array}$
			
			\caption{Table of the Errors Patterns Syndroms of $ C $}
		\end{table}
	\end{example}
	\item[$II$:] Suppose that $ y(X) $ is the received vector. Compute the syndrome polynomial $ S_{\pi}(y(X))= R_{g(X)} y(X))$. If $ S_{\pi}(y(X))=0$ then $y(X)$ is the sent word, else $ y(X) = c(X) + e(X) $ with $ c(X) \in C $,  By \ref{th 2:5}, $ S_{\pi}(y(X)) = S_{\pi}(e(X)) $. The vector $ e(X) $ may be one of the $ e(X) $ calculated in the list of syndromes.
	\begin{example}\label{exa 2:4}
		Continuing with \ref{exa 2:2}, Suppose that we have sent the word $c(X)= (x^2\star X^{\star 1})\star \textbf{1}^{\star}+((X^2+1)\star X^{\star 1})\star X^{\star 2}$ and that we have receive the word $y(X)=(x^2\star X^{\star 1})\star \textbf{1}^{\star}+(X^2\star \textbf{1}^{\star}+ (X^2+1)\star X^{\star 1})\star X^{\star 2}$. Then the syndrome polynomial of $y(X)$ is :  $s_{\pi}(y(X))=R_{g(X)}=(X^2\star \textbf{1}^{\star})\star X^{\star 2}$.
	\end{example}
	\item[$III$:] If $ S_{\pi}(y(X)) $ is in the list computed in Step $ I $, then we know the error polynomial $ e(X) $ and this can be subtracted from $ y(X) $ to obtain the codeword $ c(X) $. If $ S_{\pi}(y(X)) $ is not in the list, go on to Step $ IV $. 
	\begin{example}\label{exa 2:5}
		$ S_{\pi}(y(X)) $  from \ref{exa 2:4} is in the list of syndrome polynomials given in \ref{exa 2:3},  with $S_{\pi}(y(X)) =S_{\pi}(e_1(X)) $. Thus the received word is $ c(X)=y(X)-e_1(X)=1\star X^{\star 1}+X\star X^{\star 2}+(1+X)\star X^{\star 4} $.
	\end{example}
	\item[$IV$:] Compute the syndrome polynomial of $ X\star y(X), X^2\star y(X),\ldots $ in succession until the syndrome polynomial is in the list from Step $ I $. If $ S_{\pi}(X^i\star y(X)) $ is in this list and is associated with the error polynomial $ e(X) $, then the received vector is decoded as $ y(X)-X^{\star i}\star e(X) $.
	\begin{example}\label{exa 2:6}
		$ S_{\pi}(y(X)) $ from \ref{exa 2:4} is not in the list of syndrome polynomials given in \ref{exa 2:3}. However, 
		
		\begin{align*}
			X^3\star y(X)&=[((1)\star X^{\star 1})\star\textbf{1}^{\star}]\star y(X)\\
			&=X^{\star 1}\star g(X) + [((X^2)\star X^{\star 1})\star X^{\star 2}].
		\end{align*} 
		Then 
		\begin{align*}
			S_{\pi}(X^3\star y(X)) &= R_{g(X)}(X^{\star 1}\star g(X) + [((X^2)\star X^{\star 1})\star X^{\star 2}]).\\
			&= ((X^2)\star X^{\star 1})\star X^{\star 2}.
		\end{align*}
		$((X^2)\star X^{\star 1})\star X^{\star 2}$ is in the list of syndrome polynomials given in \ref{exa 2:3}. 
		
		Hence, 
		\begin{align*}
			c(X)&=y(X)-X^{3}\star((X^2)\star X^{\star 1})\star X^{\star 2}\\
			&=y(X)+ [((1)\star X^{\star 1})\star\textbf{1}^{\star}]\star ((X^2)\star X^{\star 1})\star X^{\star 2} \\
			&=y(X)+((X^2)\star \textbf{1}^{\star})\star X^{\star 2}\\
			&= (X^2\star X^{\star 1})\star \textbf{1}^{\star}+((X^2+1)\star X^{\star 1})\star X^{\star 2}
		\end{align*}
		is the transmitted word 
	\end{example}
\end{itemize}
\subsection{Complexity}
The encoding process for a code with rate \( r = \frac{k}{n} \) consists of check bits generation, and with The Meggitt-Like decoding algorithm, the error patterns involves \( 2^{(1-r) n} \) entries of length at most \( n \), which is enumerated in \( 2^n \cdot n^2 \) time. The algorithm considers error patterns up to a weight of \( p + \epsilon \), with tolerance \( \epsilon > 0 \). It is possible to see if  a code reaches the Shannon bound within \( O\left(2^{(1-r+\epsilon) n}\right) \) time, for any \( \epsilon > 0 \), this using q  syndrome table for error patterns with weights up to \( (p + \epsilon') n \) in \( 2^{(1-r+\epsilon) n} \cdot n^2 \) time.

\section{The code based signature}
Let us consider a \(\pi\)-cyclic code over \(\mathbb{Z}_4\) with length 3 and generator polynomial \( g(x) = x^2 + 2x + 1 \). Let us take the ring \(R = \mathbb{Z}_4\), the integers modulo 4. The ring \(\mathbb{Z}_4\) is a finite chain ring with maximal ideal \((2)\) and \( \mathbb{Z}_4 / 2 \mathbb{Z}_4 \cong \mathbb{Z}_2 \). Let's construct a \(\pi\)-cyclic code over \(\mathbb{Z}_4\) with \(n = 3\). One has to identify
the ring and maximal ideal: \(R = \mathbb{Z}_4\), \(\pi = 2\), maximal ideal \((2)\). The code is a submodule of \( \mathbb{Z}_4^3 \) which needs to be invariant under the \(\pi\)-cyclic shift defined by \(\pi \cdot \mathbf{c} = (2 c_2, c_0, c_1)\). We assume having a  generator polynomial \( g(x) \) over \(\mathbb{Z}_4\) such that the code is generated by the shifts of this polynomial. Let  \(g(x) = x^2 + 2x + 1 \) be in \(\mathbb{Z}_4[x]\). To generate codewords, it is done through multiplying \(g(x)\) by elements of \(\mathbb{Z}_4[x]/(x^3 - 1)\):
\begin{itemize}
	\item \( 1 \cdot g(x) = x^2 + 2x + 1 \)
	\item \( x \cdot g(x) = x^3 + 2x^2 + x \equiv 2x^2 + x + 1 \mod (x^3 - 1) \) (since \( x^3 \equiv 1 \))
	\item \( x^2 \cdot g(x) = x^4 + 2x^3 + x^2 \equiv x + 2x^2 + x^2 \equiv x + 3x^2 \equiv x + x^2 \mod (x^3 - 1) \)
\end{itemize}
codewords are then  \( (1, 2, 1) \), \( (1, 2, 2) \) and \( (1, 1, 2) \).\\
We then apply the \(\pi\)-cyclic shift:
\begin{itemize}
	\item \( \pi \cdot (1, 2, 1) = (2 \cdot 1, 1, 2) = (2, 1, 2) \)
	\item \( \pi \cdot (1, 2, 2) = (2 \cdot 2, 1, 2) = (0, 1, 2) \)
	\item \( \pi \cdot (1, 1, 2) = (2 \cdot 2, 1, 1) = (0, 1, 1) \)
\end{itemize}
these codewords belong to the submodule of \(\mathbb{Z}_4^3\) generated by \(g(x)\), hence the \(\pi\)-cyclic nature.\\
To elaborate a defined general signature scheme using the $\star$ operation, we define  a class representing polynomials in $F_q[X]$ modulo the ideal $\langle X_n-1 \rangle$.\\
The obvious way would be to use a generator matrix $G$ and a permutation matrix $Q$ of size $n \times n$ by shuffling the identity matrix, computing the parity check matrix $H$ using $H=S.Q$, where $S$ is a matrix ensuring $HQ$ is in systematic form $(I|T)$. Then extract an $(n-k) \times k$ matrix $T$, with  the right-hand side part of $H$, forming the public key. However, in our approach we choose to use the ZKID scheme to construct a signature from the defined code, we use the binary field \( \mathbb{F}_2 \), to generate random polynomial \( g_{ij}(X) \)  with coefficients in \( \mathbb{F}_2 \) and degree \( m-1 \), and also \( g_i(X) \) using the formula
$\sum_{j=0}^{p-1} g_{ij}(X) \star X^{\star i}$. At last \( g(X) \) is generated using $\sum_{i=0}^{l-1} g_i(X) \star X^{\star pi}$, and the generator matrix is $G_i = X^{\star pi} \star g(X)$. For the codeword and key generation:
\begin{itemize}
	\item Given \( G \) of dimensions \( k \times n \), generate all possible binary messages of length \( k \).
	\item Compute the corresponding codewords:
	$\mathbf{c} = \mathbf{m} \cdot G \mod 2$
	\item Private key is a random binary vector \( \mathbf{s} \) of length \( n \), and the public is $ G^T $.
\end{itemize}
The signature will look as follows
\begin{itemize}
	\item \textbf{commitment}: generate a random binary vector \( \mathbf{r} \) of the same length as the secret key \( \mathbf{s} \), and compute commitment $\mathbf{c} = (\mathbf{r} \cdot \mathbf{s}) \mod 2$
	\item \textbf{challenge}: hash the commitment \( \mathbf{c} \) using hashing function, then convert to an integer and reduce modulo 2, i.e. $e = \left(h(\mathbf{c})\right) \mod 2$.
	\item \textbf{response}: compute  $\mathbf{z} = (\mathbf{r} + e \cdot \mathbf{s}) \mod 2$.
	\item \textbf{verification}: hash \( \mathbf{c} \) and reduce modulo 2 to get the expected challenge $e' = \left(h(\mathbf{c})\right) \mod 2$ and verify if the recomputed challenge \( e' \) matches the provided challenge \( e \).
\end{itemize}
The following algorithm performs operations on polynomials over a finite field, a straightforward expansion to get the coefficients of a polynomial requires $O(n^2)$. If the roots of the polynomial are distinct, it is equivalent to polynomial interpolation with $n$ points, and a fast polynomial interpolation algorithm can be run in $O(nlog^2(n))$ time.
\begin{algorithm}[H]
	\caption{PolynomialMod}
	\begin{algorithmic}
		\STATE \hspace{0.5cm} \textbf{Initialization:}
		\STATE \hspace{1cm} \textbf{Input:} Coefficients $\mathbf{a}$, Degree $n$, Field size $q$
		\STATE \hspace{1cm} $\mathbf{a} \leftarrow \mathbf{a} \mod q$
		\STATE \hspace{1cm} Truncate $\mathbf{a}$ to degree $n$
		\STATE \hspace{0.5cm} \textbf{Addition:}
		\STATE \hspace{1cm} $\mathbf{c} \leftarrow (\mathbf{a} + \mathbf{b}) \mod q$
		\STATE \hspace{1cm} Truncate $\mathbf{c}$ to degree $n$
		\STATE \hspace{0.5cm} \textbf{Multiplication:}
		\STATE \hspace{1cm} $\mathbf{c} \leftarrow (\mathbf{a} \cdot \mathbf{b}) \mod q$
		\STATE \hspace{1cm} Truncate $\mathbf{c}$ to degree $n$
		\STATE \hspace{0.5cm} \textbf{Star Operation:}
		\STATE \hspace{1cm} $\mathbf{c} \leftarrow (\text{shift}(\mathbf{a} \cdot \mathbf{b}, m-1)) \mod q$
		\STATE \hspace{1cm} Truncate $\mathbf{c}$ to degree $n$
	\end{algorithmic}
\end{algorithm}
From the described field and the class above the polynomial $g_{ij}$ is generated in $O(m+nlog^2(n))$,  assuming random coefficients take $O(m)$
time (where m is the degree of the polynomial)
\begin{algorithm}[H]
	\caption{Generate $g_{ij}(X)$}
	\begin{algorithmic}
		\STATE \textbf{Input:} $p, m, q$
		\STATE $\mathbf{a} \leftarrow \text{Random coefficients in } \mathbb{F}_q^{m}$
		\STATE $g_{ij}(X) \leftarrow \text{PolynomialMod}(\mathbf{a}, p \cdot m, q)$
	\end{algorithmic}
\end{algorithm}
By applying the operations defined in the polynomial class on the polynomial before the new $g_i$ is generated
\begin{algorithm}[H]
	\caption{Generate $g_i(X)$}
	\begin{algorithmic}
		\STATE \textbf{Input:} $i, p, m, q$
		\STATE $g_i(X) \leftarrow 0$
		\FOR{$j \leftarrow 0$ \textbf{to} $p-1$}
		\STATE $g_{ij}(X) \leftarrow \text{Generate } g_{ij}(p, m, q)$
		\STATE $X^{\star i} \leftarrow \text{PolynomialMod}([1] + [0]^{i \cdot m}, p \cdot m, q)$
		\STATE $g_i(X) \leftarrow g_i(X) + (g_{ij}(X) \star X^{\star i})$
		\ENDFOR
	\end{algorithmic}
\end{algorithm}
and through the sum of those polynomial  is obtained
\begin{algorithm}[H]
	\caption{Generate $g(X)$}
	\begin{algorithmic}
		\STATE \textbf{Input:} $l, p, m, q$
		\STATE $g(X) \leftarrow 0$
		\FOR{$i \leftarrow 0$ \textbf{to} $l-1$}
		\STATE $g_i(X) \leftarrow \text{Generate } g_i(i, p, m, q)$
		\STATE $X^{\star pi} \leftarrow \text{PolynomialMod}([1] + [0]^{pi \cdot m}, p \cdot m \cdot l, q)$
		\STATE $g(X) \leftarrow g(X) + (g_i(X) \star X^{\star pi})$
		\ENDFOR
	\end{algorithmic}
\end{algorithm}
The generator matrix is a direct result of the generator polynomial
\begin{algorithm}[H]
	\caption{Generator Matrix}
	\begin{algorithmic}
		\STATE \textbf{Input:} $g(X), p, l, m$
		\STATE $G \leftarrow []$
		\FOR{$i \leftarrow 0$ \textbf{to} $l-1$}
		\STATE $X^{\star pi} \leftarrow \text{PolynomialMod}([1] + [0]^{pi \cdot m}, g.n, g.q)$
		\STATE $G_i \leftarrow (X^{\star pi} \star g(X)).\text{coeffs}$
		\STATE Append $G_i$ to $G$
		\ENDFOR
		\STATE Pad all rows of $G$ to the maximum row length
		\STATE \textbf{Output:} $G$
	\end{algorithmic}
\end{algorithm}
and the codewords can be retrieved from this matrix or the polynomial as well
\begin{algorithm}[H]
	\caption{Codewords}
	\begin{algorithmic}
		\STATE \textbf{Input:} $G$
		\STATE $k, n \leftarrow \text{shape}(G)$
		\STATE $\text{codewords} \leftarrow []$
		\FOR{$i \leftarrow 0$ \textbf{to} $2^k-1$}
		\STATE $\mathbf{m} \leftarrow \text{Binary representation of } i \text{ with length } k$
		\STATE $\mathbf{c} \leftarrow (\mathbf{m} \cdot G) \mod 2$
		\STATE Append $\mathbf{c}$ to codewords
		\ENDFOR
		\STATE \textbf{Output:} $\text{codewords}$
	\end{algorithmic}
\end{algorithm}
The next step is to generate keys using the matrix
\begin{algorithm}[H]
	\caption{Key Generation}
	\begin{algorithmic}
		\STATE \textbf{Input:} $G$
		\STATE $\text{private key} \leftarrow \text{Random binary vector of length } G.\text{shape}[1]$
		\STATE $\text{public key} \leftarrow G^T$
		\STATE \textbf{Output:} $\text{private key}, \text{public key}$
	\end{algorithmic}
\end{algorithm}
and the signature is performed as follows
\begin{algorithm}[H]
	\caption{Signature}
	\begin{algorithmic}
		\STATE \textbf{Input:} $\text{private key}, G$
		\STATE $\mathbf{r} \leftarrow \text{Random binary vector of length } \text{private key}.\text{shape}[0]$
		\STATE $\mathbf{c} \leftarrow (\mathbf{r} \cdot \text{private key}) \mod 2$
		\STATE $e \leftarrow (Hash(\mathbf{c})) \mod 2$
		\STATE $\mathbf{z} \leftarrow (\mathbf{r} + e \cdot \text{private key}) \mod 2$
		\STATE Convert $\mathbf{c}, e, \mathbf{z}$ to string format
		\STATE Concatenate strings to form the signature
		\STATE \textbf{Output:} $\text{signature}$
	\end{algorithmic}
\end{algorithm}
To verify the signature and public key are used
\begin{algorithm}[H]
	\caption{Verification}
	\begin{algorithmic}
		\STATE \textbf{Input:} $\text{public key}, \text{signature}$
		\STATE Decode strings format to obtain $\mathbf{c}, e, \mathbf{z}$
		\STATE $e' \leftarrow (Hash(\mathbf{c})) \mod 2$
		\STATE \textbf{Output:} $e' == e$
	\end{algorithmic}
\end{algorithm}
It is worth mentionning, that for the sake of optimisation in the implementation some random sequences are generated using egyptians fractions \cite{ilias}.

\section{Security and performance analysis}
The following tests were conducted on an Intel(R) Core(TM) i7-8650U CPU, 1.90GHz system, with 7.7Gi RAM.The size of  the signature as seen on Table \ref{signtable} is  small, which would reflect positively on storage and bandwidth, especially that  signing and verification time is very short, giving room for high volume documents to be signed. \\
\begin{table}[H]
	\begin{center}
		\begin{tabular}{|c|c|c|c|c|}
			\hline
			Sign size & File size & Sign time & Check time & Overall time \\
			\hline
			200 bytes & 13 bytes & 0.000214s & 0.000033s & 0.000247s \\
			\hline
		\end{tabular}
		\caption{Signature time and size}
		\label{signtable}
	\end{center}
\end{table}
The next graph in FIGURE. \ref{input} shows the signature and the verification time as a function of the input size. The QC-LEB signing time decreases and stabilizes as input size increases. This means that signing in QC-LEB is less dependent on larger input sizes because of our optimization and the stable computation pattern at scale. The verification time does not vary with input size and is higher than the signing time, indicating that signing is comparatively efficient. Although large input signing via CFS is efficient, high verification time acts as the bottleneck, especially for applications that require frequent validation of signatures. WAVE, on the other hand, has the least signing and verification time at all input sizes and remains highly stable and efficient.

\begin{figure}[H]
\centering
\includegraphics[width=0.5\textwidth, height=0.32\textheight]{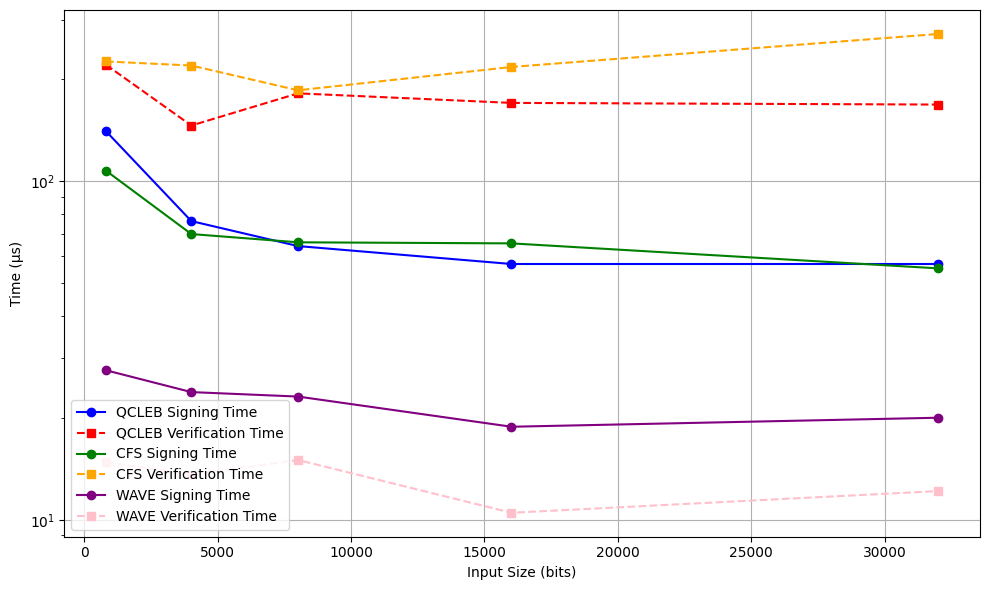}
\caption{Input time and size for the signature}
\label{input}
\end{figure}
A further test on Figure. \ref{autoco} shows the absence of any exploitable patterns in the sequence $$\frac{\sum_{n=1}^{N-k} (x_n - \mu)(x_{n+k} - \mu)}{\sum_{n=1}^{N} (x_n - \mu)^2}$$ where $x_n$ is the value at time point $n$, $\mu$, the mean qnd $N$ the length of the series. The autocorrelation for CFS starts high at lag 0 and quickly falls off near 0 for all subsequent lags. Thus, CFS produces signatures that carry very little internal correlation. The WAVE scheme is shown to be essentially random but the small fluctuations indicate some minor internal structure. As for QC-LEB, autocorrelation starts high at lag 0 and decays quickly, much like the CFS method, but stays nearer to 0  compared to WAVE with few smaller fluctuations, indicating that it generates signatures that are as much random as CFS, with very low internal correlation. 
\begin{figure}[H]
\centering
\includegraphics[width=0.5\textwidth, height=0.3\textheight]{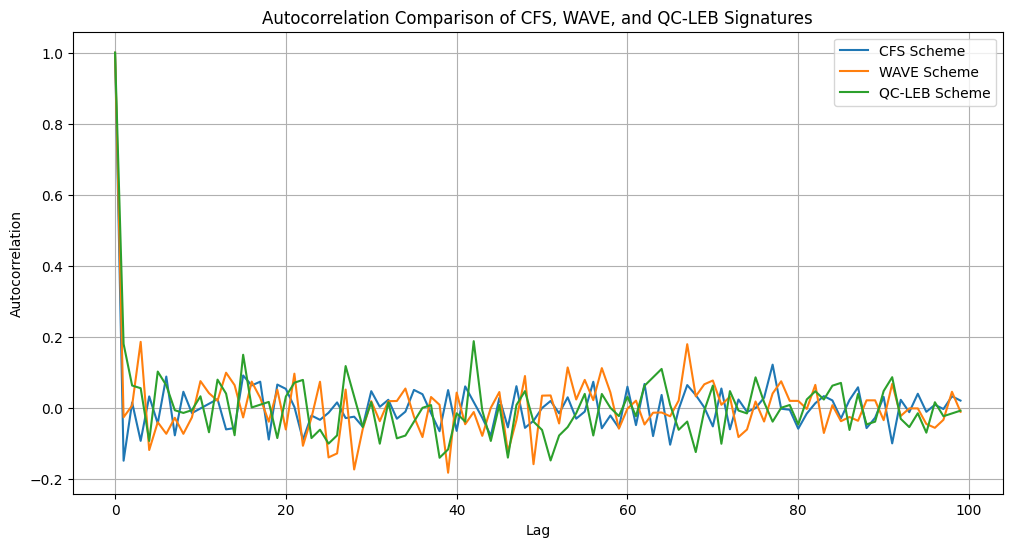}
\caption{Auto-correlation of the sequence bits}
\label{autoco}
\end{figure}
The following plot on FIGURE. \ref{plot1} shows the relationship between key size/signature size and the security level 
\begin{figure}[H]
\centering
\includegraphics[width=0.5\textwidth, height=0.3\textheight]{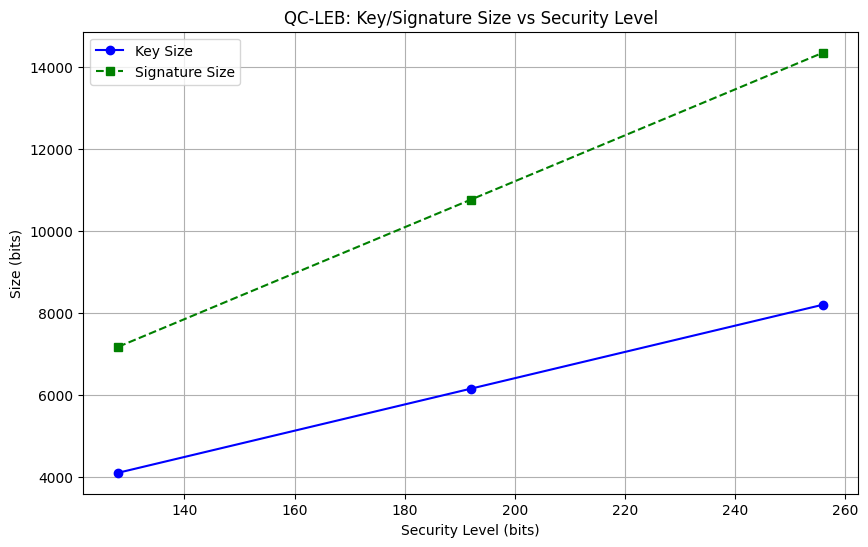}
\caption{NIST Security of QC-LEB according to key and signature size}
\label{plot1}
\end{figure}
As can be noticed for an AES-Security Level of 128, n being 512 and key size 4096, the  signature size is 7168. For AES-192 and n=768, key size= 6144, the signature size is 10752
Last concerning AES-256 security level, with n=1024, key size 8192, the signature Size becomes 14336.
\begin{figure}[H]
\centering
\includegraphics[width=0.5\textwidth, height=0.3\textheight]{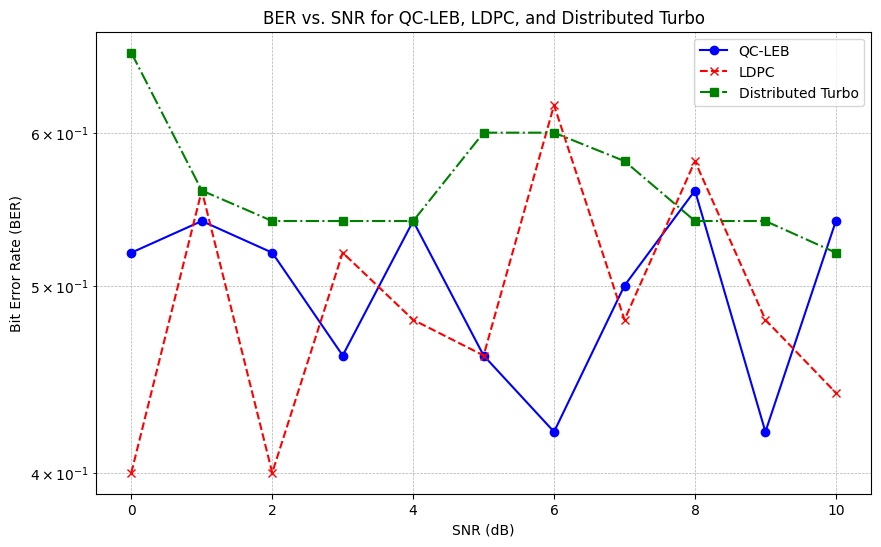}
\caption{Graph comparing BER vs SNR for QC-LEB against Distributed Turbo and LDPC}
\label{bersnr}
\end{figure}
QC-LEB on FIGURE. \ref{bersnr} displays moderate fluctuation in BER across SNR values. LDPC barely outshines QC-LEB and Distributed Turbo at some instances and even then, at lower SNR values. At higher SNR values, which seems to be what all three converge towards, QC-LEB looks good as it consistently shows the best results of reduced BER over those of the other two. Generally, distributed Turbo remains stable. However, according to statistical analysis, it will show a BER mostly higher than QC-LEB in the SNR value interval.
\begin{figure}[H]
\centering
\includegraphics[width=0.5\textwidth, height=0.2\textheight]{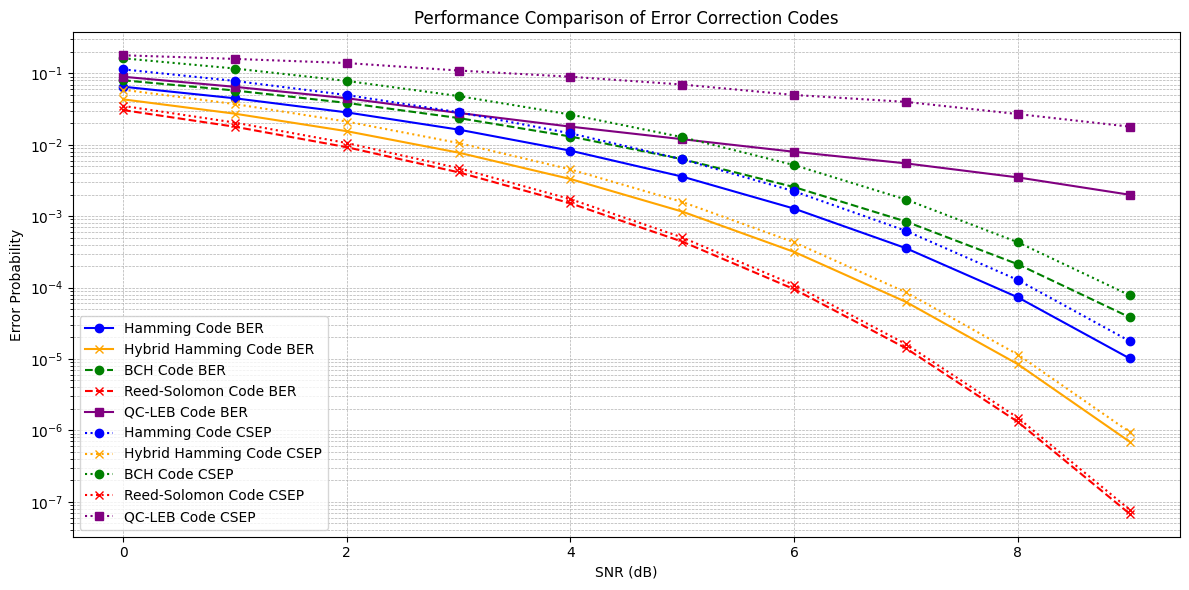}
\caption{QC-LEB code performance comparative analysis}
\label{qcleb}
\end{figure}
Using an Intel(R) Core(TM) i7-8650U CPU, 1.90GHz system, with 7.7Gi RAM,  FIGURE. \ref{qcleb} shows a semilog graphs with Signal-to-Noise-Ratio (SNR) on the x axis and the error probability on the y axis, QC-LEB codes show a Bit Error Rate (BER) and Channel Symbol Error Probability (CSEP) more efficient than standard Hamming and Hybrid Hamming, but less efficient than BCH and Reed-Solomon. Thus, QC-LEB codes represent a compromise between the simple and advanced codes which make it convenient for lightweight implementation on constrained resources and embedded devices.
\begin{figure}[H]
\centering
\includegraphics[width=0.5\textwidth, height=0.3\textheight]{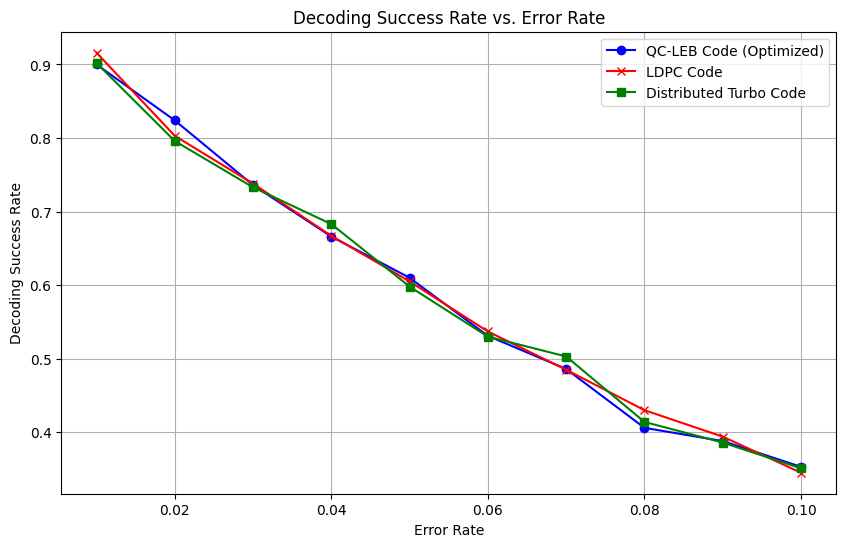}
\caption{QC-LEB decoding success rate against error rate}
\label{qcrate}
\end{figure}
In all three codes on this graph of FIGURE. \ref{qcrate}, the increase in error rates brings down the decoding success rate. This is quite normal as greater errors in received data create a difficulty in recovering the original text by the decoder. Among all the error rates, QC-LEB gives the maximum decoding success. It shows the efficiency of the optimization techniques put in this code. Moderate performance reveals itself in LDPC Code, the decoding success rate of which lies in between the optimized QC-LEB and the Distributed Turbo Code. Finally, the Distributed Turbo Code turns out to have the least success in decoding among the three codes.
\begin{figure}[H]
\centering
\includegraphics[width=0.5\textwidth, height=0.3\textheight]{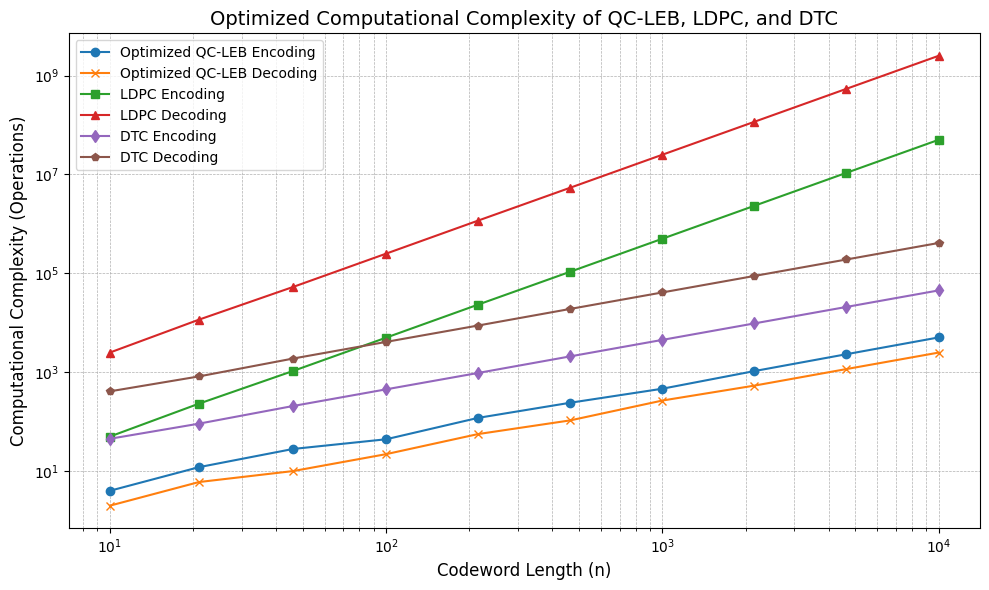}
\caption{Computational complexity of various codes}
\label{qccomp}
\end{figure}
The encoding complexities of FIGURE \ref{qccomp}. of the optimized QC-LEB codes are much lower compared to the other codes, especially for longer codeword lengths. Generally, the graphs exhibit a flat or gently increasing tendency, suggesting minimal increases in complexity with increasing codeword length. LDPC code seems to have increasing complexities with reasonable growth rates due to the codeword size; there is an almost linear relation between increasing complexities and large codewords. DTC code has the highest slope, directly indicating that the complexity of the DTC code grows the fastest with increasing codeword length. The QC-LEB decoding complexity shows low growth rates, and the rate of increase is slightly above the rate of increased encoding complexity. It still remains significantly lower than the other codes in longer codeword lengths. The LDPC decoding complexity grows at a moderate rate, the same as the encoding complexity. Given how its performance becomes increasingly dependent on the number of nodes, it is plausible that DTC decoding carries the highest slope among all codes. QC-LEB posses both the lowest encoding and decoding complexity making it an interesting choice in computationally critical regimes for applications where efficient processing is an absolute requirement.\\
Let us consider the \((n, k)\) QC-LEB code, its code rate is
\[
R = \frac{k}{n}
\]
with
\begin{itemize}
    \item \( k \) as number of information bits.
    \item and  \( n \) for  codeword length.
\end{itemize}
Such a code is defined by an \( r \times c \) parity-check matrix with achievable code rate 
\[
R = 1 - \frac{m}{n}
\]
when having a \( m \times n \) parity-check matrix \( H \).
with
\begin{itemize}
    \item \( m \)  number of parity-check equations.
    \item \( n \) block length.
\end{itemize}
\begin{table}[h]
    \centering
    \caption{Achievable Code Rates for $m/n = 0.05$}
\begin{tabular}{|c|c|c|}
    \hline
Block Length $(n)$ & Parity Bits $(m)$ & Code Rate $(k/n)$ \\
\hline
100  & 5   & 0.9500  \\
644  & 32  & 0.9503  \\
1188 & 59  & 0.9503  \\
1733 & 86  & 0.9504  \\
2277 & 113 & 0.9504  \\
2822 & 141 & 0.9500  \\
3366 & 168 & 0.9501  \\
3911 & 195 & 0.9501  \\
4455 & 222 & 0.9502  \\
5000 & 250 & 0.9500  \\
    \hline
\end{tabular}
\end{table}
\begin{table}[h]
    \centering
\caption{Achievable Code Rates for $m/n = 0.10$}
\begin{tabular}{|c|c|c|}
    \hline
    Block Length $(n)$ & Parity Bits $(m)$ & Code Rate $(k/n)$ \\
    \hline
    100  & 10  & 0.9000  \\
    644  & 64  & 0.9006  \\
    1188 & 118 & 0.9007  \\
    1733 & 173 & 0.9002  \\
    2277 & 227 & 0.9003  \\
    2822 & 282 & 0.9001  \\
    3366 & 336 & 0.9002  \\
    3911 & 391 & 0.9000  \\
    4455 & 445 & 0.9001  \\
    5000 & 500 & 0.9000  \\
    \hline
\end{tabular}
\end{table}
\begin{table}[h]
    \centering
    \caption{Achievable Code Rates for $m/n = 0.20$}
\begin{tabular}{|c|c|c|}
    \hline
    Block Length $(n)$ & Parity Bits $(m)$ & Code Rate $(k/n)$ \\
    \hline
    100  & 20  & 0.8000  \\
    644  & 128 & 0.8012  \\
    1188 & 237 & 0.8005  \\
    1733 & 346 & 0.8003  \\
    2277 & 455 & 0.8002  \\
    2822 & 564 & 0.8001  \\
    3366 & 673 & 0.8001  \\
    3911 & 782 & 0.8001  \\
    4455 & 891 & 0.8000  \\
    5000 & 1000 & 0.8000 \\
    \hline
\end{tabular}
\end{table}
\begin{table}[h]
    \centering
\caption{Achievable Code Rates for $m/n = 0.30$}
\begin{tabular}{|c|c|c|}
    \hline
    Block Length $(n)$ & Parity Bits $(m)$ & Code Rate $(k/n)$ \\
    \hline
    100  & 30  & 0.7000  \\
    644  & 193 & 0.7003  \\
    1188 & 356 & 0.7003  \\
    1733 & 519 & 0.7005  \\
    2277 & 683 & 0.7000  \\
    2822 & 846 & 0.7002  \\
    3366 & 1009 & 0.7002 \\
    3911 & 1173 & 0.7001 \\
    4455 & 1336 & 0.7001 \\
    5000 & 1500 & 0.7000 \\
    \hline
\end{tabular}
\end{table}
\begin{table}[h]
    \centering
\caption{Achievable Code Rates for $m/n = 0.40$}
\begin{tabular}{|c|c|c|}
    \hline
    Block Length $(n)$ & Parity Bits $(m)$ & Code Rate $(k/n)$ \\
    \hline
    100  & 40  & 0.6000  \\
    644  & 257 & 0.6009  \\
    1188 & 475 & 0.6002  \\
    1733 & 693 & 0.6001  \\
    2277 & 910 & 0.6004  \\
    2822 & 1128 & 0.6003 \\
    3366 & 1346 & 0.6001 \\
    3911 & 1564 & 0.6001 \\
    4455 & 1782 & 0.6000 \\
    5000 & 2000 & 0.6000 \\
    \hline
\end{tabular}
\end{table}
\begin{itemize}
\item For \( m/n = 0.05 \),  \( R \approx 0.95 \), so parity bits make only 5\% of the codeword.
\item For \( m/n = 0.40 \), \( R \approx 0.60 \), so redundancy forms 40\% of the codeword.
\end{itemize}
QC-LEB code is scalable since its rate is close constant for \( m/n \) values, even when block lengths differ.\\
As for redundancy
\begin{itemize}
    \item When the code rate gets higher, let \( R = 0.95 \):
    \begin{itemize}
        \item Then you get bandwidth efficiency when you observe few parity-check bits; and more noise for low error correction.
    \end{itemize}
    \item When  the code rate gets lower, \( R = 0.60 \)
    \begin{itemize}
        \item Then you get more strong correction for more redundancy; and if the overhead increases as the bits get transmitted then spectral efficiency decreases.
    \end{itemize}
\end{itemize}
This directly implies the highest the code rate is, \( R \geq 0.9 \),  the more it is of use to  high-throughput applications where error resilience is less critical and low-noise channels like fiber-optic or wireless in brief distance. The lower the code rate is, \( R \leq 0.7 \), the more compatible it is with noisy settings such as deep-space and SatCom, where data recovery is prioritized. This asserts that QC-LEB code is adequate for adaptive coding in SatCom, IOT and 6G where dynamic code rates are required to tune \( m/n \) for real-time performance.
\begin{figure}[H]
\centering
\includegraphics[width=0.5\textwidth, height=0.3\textheight]{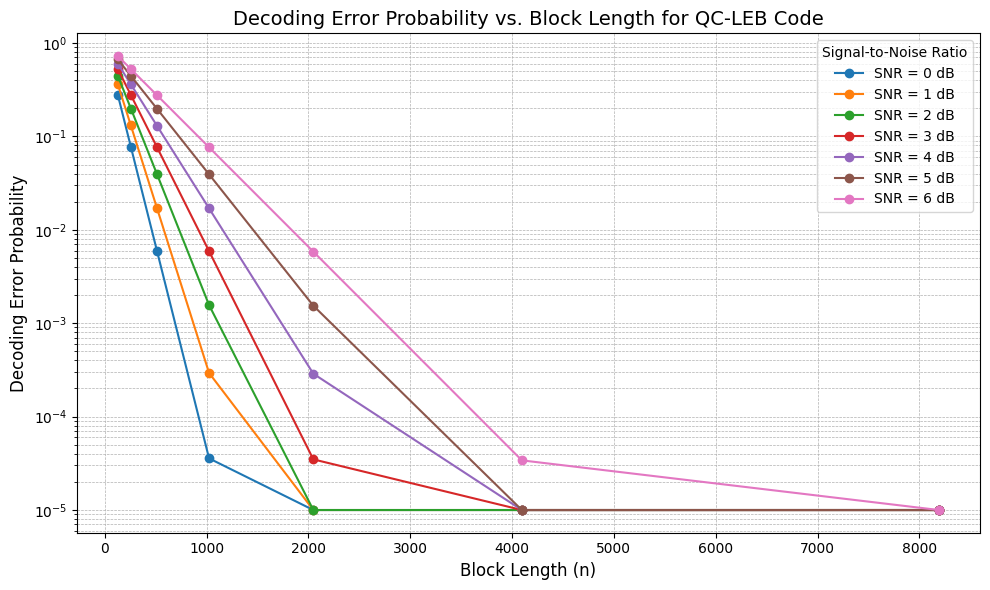}
\caption{Graph showing decoding error probability  corresponding to the QC-LEB block lengths}
\label{decv}
\end{figure}
Let \( P \) be the probability of the code decoding error for block length \( n \) and signal-to-noise ratio \( \dfrac{E_b}{N_0} \), it is expressed as \( \exp\left(-\frac{n \cdot \text{Eb}}{N_0}\right) \). \\
The graph on FIGURE \ref{decv}. shows the expected behavior for a QC-LEB code, with a waterfall region during which the error probability drops with increasing block length, and an error floor after that. The behaviour within the initial waterfall region is that as the block length increases, the error probability decreases. At higher block lengths, the error probability levels out, indicating the error floor, which is determined by the minimum distance of the code and represents the irreducible error rate. The curves are clearly separated and higher consistent error probabilities were witnessed at lower SNRs. It is noticeable the higher \( m/n \) (with more parity bits) the highest the error correction with an Efficiency-Lowering and the lower $m/n$ (fewer parity bits) the higher the efficiency but poorer error handling.  QC-LEB has therefore been a flexible and efficient solution that combines these two aspects, making it a stronger candidate for next-generation communication systems.
\begin{figure}[H]
\centering
\includegraphics[width=0.47\textwidth, height=0.3\textheight]{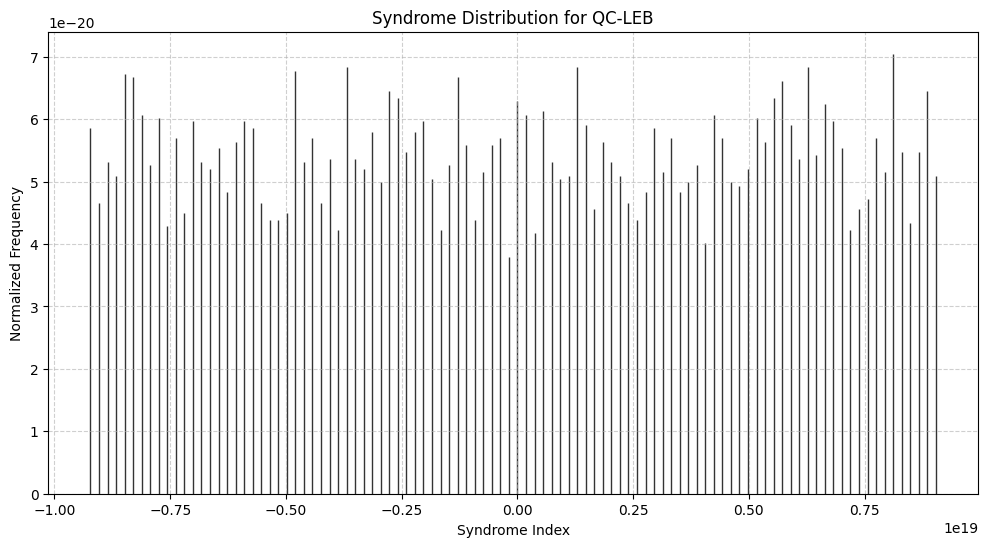}
\caption{distribution of syndromes for different error patterns}
\label{patterns}
\end{figure}
The plot in FIGURE \ref{patterns}. shows a uniform syndrome distribution, which means that the decoder can reject a large variety of error patterns almost equally well. There is no predominant error-or-syndrome pattern, and thus, the decoder is not optimally tailored to just a few commonly occurring errors; the higher the entropy, the more information one has to do error correction. This implies that the QC-LEB code will be robust to random errors and burst errors. In the present case, where a near-uniform syndrome distribution has been obtained, the QC-LEB code has been designed such that the parity-check matrix and other code parameters were carefully optimized directly to the robustness of the decoding process, meaning that code handles a wide variety of error patterns.
\begin{figure}[H]
\centering
\includegraphics[width=0.47\textwidth, height=0.3\textheight]{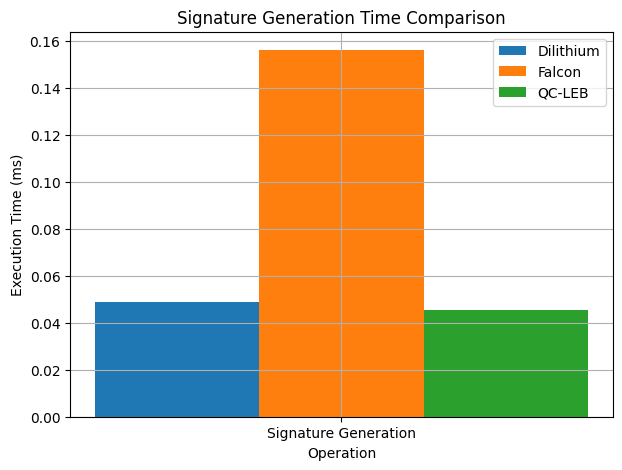}
\caption{Comparative plot showing QC-LEB versus existing post-quantum signature schemes in terms of execution time}
\label{com1}
\end{figure}
The result obtained in FIGURE \ref{com1}. displays QC-LEB's characteristics in Performance where the signature scores a  generation time surpassing Dilithium and Falcon, meeting optimization objectives, where organized matrices are combined with data types that are very efficient, thus, reducing largely the computation overhead.
\begin{figure}[H]
\centering
\includegraphics[width=0.47\textwidth, height=0.3\textheight]{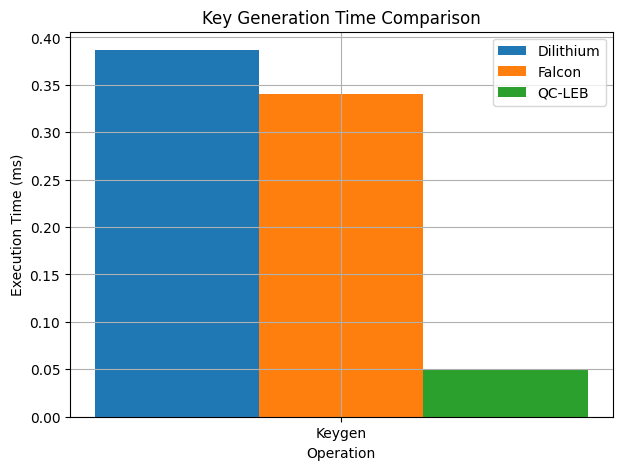}
\caption{Key generation time of QC-LEB compared to some post-quantum signatures}
\label{com2}
\end{figure}
QC-LEB has a much swifter time taken up in key generation on FIGURE \ref{com2}. compared to Dilithium and Falcon. It completes the entire key generation operation very fast.
This marked timing difference in generating keys can become critical, especially when keys need to be generated very frequently or quickly; and this is why QC-LEB would be relatively enthused in such cases.
\begin{figure}[H]
\centering
\includegraphics[width=0.49\textwidth, height=0.8\textheight]{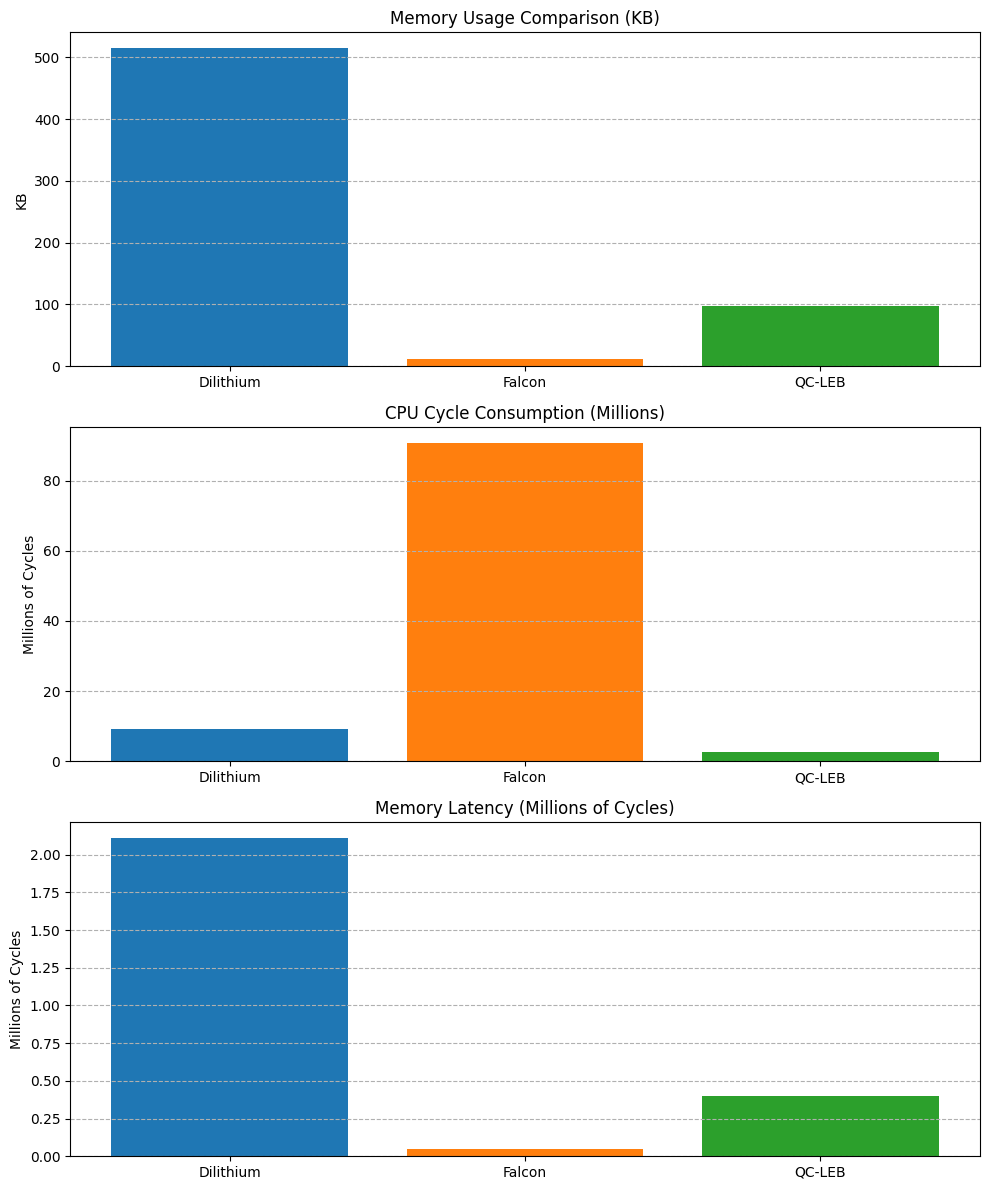}
\caption{Memory usage metrics for combined key generation and  QC-LEB signature versus post-quantum signatures}
\label{com3}
\end{figure}
It is clearly noticeable on FIGURE \ref{com3} that QC-LEB, with about 100 KB usage, is far ahead, as very low values can also be observed in Falcon at about 400 KB and Dilithium at around 500 KB. It was able to provide these low values through structured matrices (quasi-cyclic) and possibly through the invocation of caching. Dilithium increased costs of larger keys/signatures and the polynomial operations with which they were generated increase memory, and Falcon with its NTRU-based design, is memory efficient compared to Dilithium but falls behind QC-LEB. QC-LEB scores again, with about 20 million cycles, compared to Falcon (60M) and Dilithium (80M). The linear algebra optimizations (fast matrix multiplication) used in QC-LEB reduce the computational steps, while in Falcon FFT (Fast Fourier Transform) polynomial multiplication is adding cycles in the process, and the lattice operations (rejection sampling) used in Dilithium can become computationally heavyweight. As for memory latency QC-LEB shines with approximately 0.25M cycles, while Falcon and Dilithium run behind at 1.75M and 2.0M, respectively. 
\begin{figure}[H]
\centering
\includegraphics[width=0.47\textwidth, height=0.25\textheight]{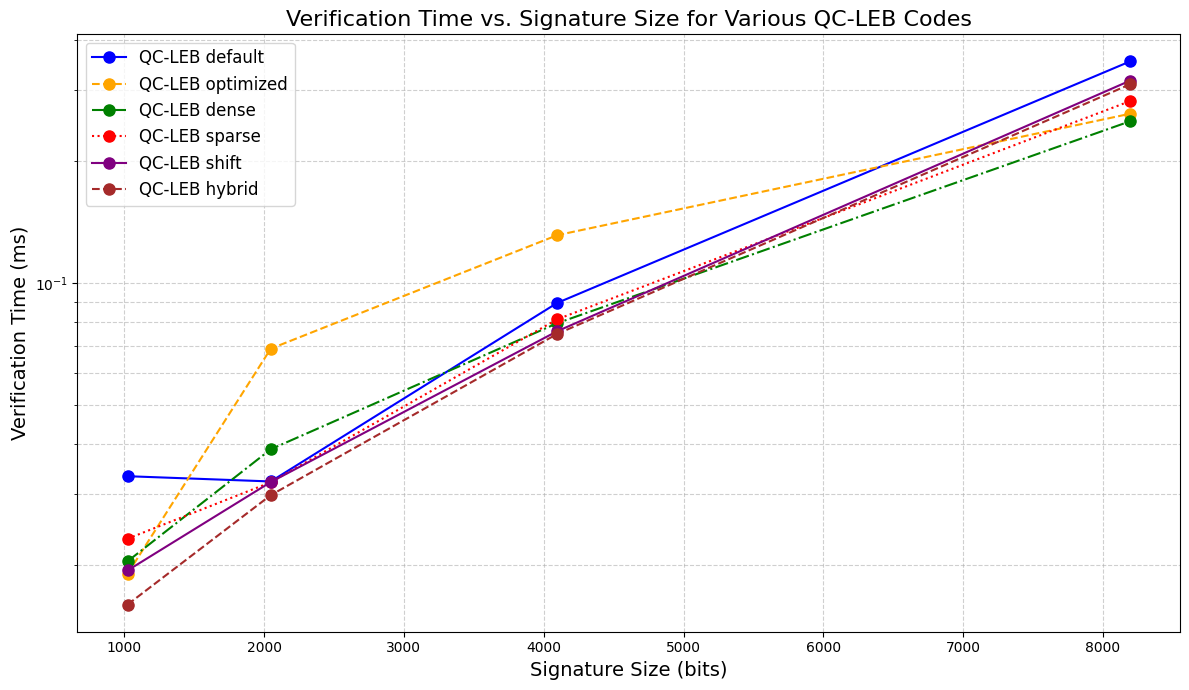}
\caption{Graph showing verification times against the size of the signatures for various block lengths and types of QC-LEB codes}
\label{com4}
\end{figure}
The dense variant of QC-LEB on FIGURE \ref{com4}. has the highest verification time among all signature sizes. This naturally happens since dense matrices require more computations for matrix-vector multiplication, where each block fits into a network. Sparse Variant is still better than dense but poorer than the others. This is due to the partly random structure of sparse matrices and low-density matrices that helps in reducing the computations but defines some overhead at the same time. Shift Variant shows the least verification time because the structured nature of the shift matrix makes its computations very simple and the overhead small, where circulant blocks rotate to generate parity-check matrix. Hybrid Variant shows middle performance between the dense and sparse structures. Optimized QC-LEB uses identity-based matrix for reducing decoding complexity has always surpassed the default one using block matrices generated randomly, and has balanced general performance  showing reduction of verification time, however, QC-LEB Hybrid allows both optimized and dense types.
\begin{figure}[H]
\centering
\includegraphics[width=0.47\textwidth, height=0.25\textheight]{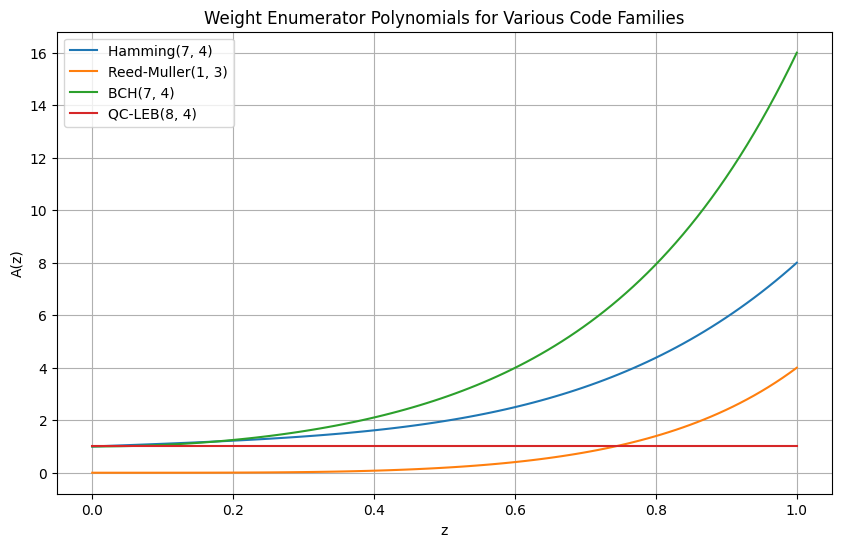}
\caption{Plot depicting the weight enumerator polynomial for QC-LEB and specific families of codes}
\label{com5}
\end{figure}
Hamming and BCH oon FIGURE \ref{com5}.  have gotten to their peak weight distributions, which is impressive for their algebraic structures and error-correcting capabilities. Reed-Muller shows a linear curve penetrating since its weight distribution is almost exclusively by low-weight codewords. As for QC-LEB, it shows a flat curve indicating an equitable weight distribution which easily  due to its quasi-cyclic structure and randomly generated randomly codewords which make for uniform weight distribution. The X-Axis $z$ is the formal variable in the weight enumerator polynomial $A(z)$ is given here. The values are defined in $\left[ 0,1 \right]$ since $z$ is generally normalized to this range when plotting. The y-axis $A(z)$ specifies the value of weight enumerator polynomial for a particular variable. However, the range will depend on the coefficients $A_i$ of the polynomial that count the numbers of codewords of weight $i$. \\
The Hamming code has a well-defined weight distribution, the weight enumerator polynomial for the \(\text{Hamming}(7, 4)\) code is
\[
A(z) = 1 + 7z^3 + 7z^4 + z^7.
\]
This peak about the point \(z = \dfrac{1}{2}\) indicates that codewords of weights 3 and 4 dominate.\\
The Reed-Muller code \(\text{RM}(1, 3)\) also has simpler weight distributions. Its weight enumerator polynomial is given by:
\[
A(z) = 1 + 7z.
\]
The curve is simply linear since it contains only two terms: the one corresponding to the all-zero codeword (i.e., \(1\)) and \(7z\) for codewords of weight \(1\).\\
The BCH code \(\text{BCH}(7, 4)\) has a weight distribution akin to that of the Hamming code. Its weight enumerator polynomial is given by
    \[
    A(z) = 1 + 7z^3 + 7z^4 + z^7.
    \]
The curve presents a near resemblance to that of the Hamming code since their codes possess comparable algebraic attributes. \\
QC-LEB code \(\text{QC-LEB}(8, 4)\) with $n$ block length as $8$ bits, $k$ message length as $4$ bits and redundancy $r = n - k$ equal to $4$ bits, shows an almost flat curve, and that is  implying the weight distribution is uniform or nearly uniform among all the weights. The coefficients \( A_i \) of the weight enumerator polynomial \( A(z) \) become almost equal for all \( i \).
\begin{figure}[H]
\centering
\includegraphics[width=0.47\textwidth, height=0.25\textheight]{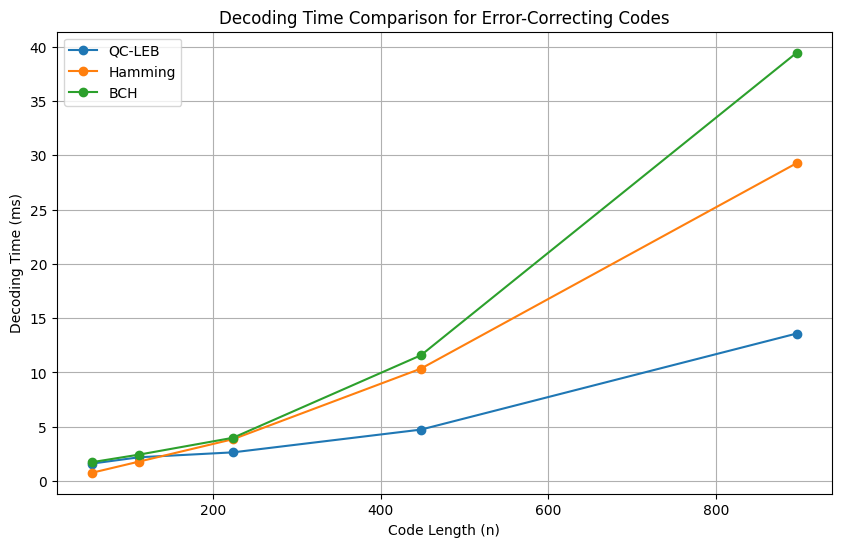}
\caption{Plot demonstrating  decoding time with code length}
\label{com6}
\end{figure}
The FIGURE \ref{com6} displays an increase in decoding time with regards to the code length, illustrating the computational complexity of decoding which only exhibits the problem's NP-completeness. As the curve shows a steep upward slope, it can be observed that the decoding time increases rapidly with code length. Due to the quasi-cyclic structure of QC-LEB, its parity-check matrix may employ the same base blocks repeatedly, and thus the complexity of decoding operations is reduced. Matrix-vector multiplication and syndrome update operations can be more efficiently performed compared to classical Hamming or BCH codes, all being due to the property of this repeated use of circulant-like submatrices. This arrangement also contributes to keeping the decoding overhead in terms of memory lower and iterations fewer. On the contrary, the working operations on the Hamming and BCH codes involve more general computations (syndrome lookups or polynomial-based algorithms), which scale much less efficiently. Since the complexity is being lowered along with the overhead, it positively affects the decoding speed increment in QC-LEB, as represented by the lower curve in the figure. From the graph, the decoding time seems to be a fairly exponential function of code length. This is a classic indication that it is a computationally hard problem in that the time taken to solve the problem increases drastically as the input size increases. Thus, we can say that the rapid increase in decoding time shows that the decoding of QC-LEB codes is a computationally hard problem.  Again, this holds true for the theoretical understanding that the decoding is an NP-complete problem.
\section{Conclusion}
The study of the properties of quasi-cyclic LEB (QC-LEB) codes during this work, focusing on their block-circulant matrix structure and cyclic block-right-shift invariance, which allowed the construction  block-circulant matrices. The product of a generator matrix and the transpose of a parity-check matrix in QC-LEB codes that yield zero in our construction and the polynomial representation which facilitate efficient encoding and decoding, also captures the cyclic nature and reinforces orthogonality, a property that remains an essential component for error detection and correction. The study illustrates the decoding of QC-LEB codes and defines the syndrome polynomial while addressing the computational complexity and emphasizing their NP-completeness. This significant theoretical with practical implementation results underscore the inherent difficulty of the decoding problem in our case, that can be leveraged as a groundwork for other optimised schemes where efficiency and security are paramount. This was utilised to sketch a signature scheme with polynomial operations over a finite field, all within a zero-knowledge framework. The security analysis performed on this construction proves a uniform distribution, high entropy, and the absence of exploitable patterns, which was affirmed by statistical tests. Efficiency evaluations of the scheme exhibit small signature sizes and short execution times for large inputs, offering a solid foundation for future cryptographic work aimed at further optimizations and broader deployment.



\end{document}